\documentclass{article}
\usepackage[utf8]{inputenc}
\usepackage[T1]{fontenc}

\pdfoutput=1
\usepackage{kr}

\usepackage[english]{babel}

\usepackage{times}
\usepackage{soul}
\usepackage{url}
\usepackage[hidelinks]{hyperref}
\usepackage[utf8]{inputenc}
\usepackage[small]{caption}
\usepackage{graphicx}
\usepackage{amsmath}
\usepackage{amsthm}
\usepackage{amsfonts}
\usepackage{booktabs}
\usepackage{algorithm}
\usepackage{algorithmic}
\usepackage{amssymb}
\usepackage{mathtools}
\usepackage{balance}

\newif\ifextended
\extendedtrue

\ifextended
  \usepackage[bibliography=common]{apxproof}
\else
  \usepackage[bibliography=common,appendix=strip]{apxproof}
\fi

\usepackage[switch,mathlines]{lineno}
\usepackage[textwidth=13mm,tickmarkheight=1em,disable]{todonotes}

\usepackage{natbib}

\usepackage[capitalise,nameinlink]{cleveref}

\usetikzlibrary{automata}
\usetikzlibrary{backgrounds}
\tikzstyle{st}=[draw,thick,circle,font=\scriptsize]
\tikzstyle{lbl}=[font=\scriptsize,auto]

\newcommand{\unionbuchi}{\sqcup}
\newcommand{\intersectbuchi}{\sqcap}

\newtheoremrep{theorem}{Theorem}
\newtheorem{example}[theorem]{Example}
\newtheoremrep{lemma}[theorem]{Lemma}
\newtheoremrep{proposition}[theorem]{Proposition}
\newtheoremrep{definition}[theorem]{Definition}
\newtheoremrep{observation}[theorem]{Observation}
\newtheoremrep{corollary}[theorem]{Corollary}

\title{Effective AGM Belief Contraction: A Journey beyond the Finitary Realm\ifextended\\(Technical Report)\fi}
\hypersetup{
  pdftitle={Effective AGM Belief Contraction: A Journey beyond the Finitary Realm}
}

\newif\ifanonymous
\anonymousfalse

\ifanonymous
\author{(Anonymous Authors)}
\else
\author{%
Dominik Klumpp$^{1,2}$\and
Jandson S. Ribeiro$^3$
\affiliations
$^1$University of Freiburg, Germany\\
$^2$LIX - CNRS - \'Ecole Polytechnique, France\\
$^3$Cardiff University, United Kingdom\\
\emails
klumpp@lix.polytechnique.fr,
ribeiroj@cardiff.ac.uk
}
\hypersetup{
  pdfauthor={Dominik Klumpp, Jandson S. Ribeiro}
}
\fi

\newcommand\prop[1]{\textbf{(#1)}}

\newcommand\powerset[1]{\mathcal{P}(#1)}
\newcommand\N{\mathbb{N}}

\newcommand\LL{\mathbb{L}}
\newcommand\Fm{\mathit{Fm}}
\newcommand\CnFun{\mathit{Cn}}
\newcommand\Cn[1]{\CnFun(#1)}
\newcommand\CCT[1][]{\mathit{C\!C\!T}_{\!#1}}
\newcommand\kb{\mathcal{K}}
\newcommand\theories[1][]{\mathsf{Th}_{#1}}

\newcommand\limplies\to

\newcommand\compl[1]{\overline{\omega}(#1)}

\newcommand\exc{\mathbb{E}}

\newcommand\AP{\mathit{AP}}
\newcommand\LTL{\mathit{L\!T\!L}}

\newcommand{\Next}[1][]{\mathbf{X}^{#1}\,}
\newcommand\Globally{\mathbf{G}\,}
\newcommand\Finally{\mathbf{F}\,}
\newcommand\Until{\mathbin{\mathbf{U}}}

\newcommand\id[1]{\mathit{id}(#1)}
\newcommand\upCCTFun{\mathit{Th}_\UP}
\newcommand\upCCT[1]{\upCCTFun(#1)}

\newcommand\UP{\mathit{UP}}
\newcommand{\lang}[1]{\mathcal{L}(#1)}

\newcommand\Traces[1]{\mathit{Traces}(#1)}
\newcommand\ltlbuchi[1]{A_{#1}}
\newcommand\kripkebuchi[1]{A_{#1}}

\newcommand\support[1]{\mathcal{S}(#1)}
\newcommand\buchi{\textnormal{B\"uchi}}
\newcommand\excBuchi{\exc_\buchi}
\newcommand\alphBuchi{\Sigma_\buchi}
\newcommand\encBuchi{f_\buchi}

\newcommand{\relFn}{\mathcal{R}}
\newcommand{\rel}[1]{{\relFn(#1)}}

\newcommand\maxaut[2]{A_{\max}^{#1,#2}}

\newcommand{\cp}[1][1]{$\mathbf{(K^{-}_{#1})}$}

\makeatletter
\newcommand{\oset}[3][0ex]{%
  \mathbin{\mathop{#3}\limits^{
    \vbox to#1{\kern-2\ex@
    \hbox{$\scriptstyle#2$}\vss}}}}
\makeatother
\newcommand\dotmin[1][]{\oset[-.7ex]{\textbf{\large.}}{-}_{#1}}

\newcommand\cleav{\mathcal{C}}

\newcommand{\minCleav}[2]{{\min}_{#1}(#2)}

\newcommand\decomp[1]{\mathit{decomp}(#1)}
\newcommand{\dotdiv}{{\dotmin}}

\begin{document}

\maketitle

\begin{abstract}
Despite significant efforts towards extending the AGM paradigm of belief change beyond finitary logics, the computational aspects of AGM have remained almost untouched.
We investigate the computability of AGM contraction on non-finitary logics, and show an intriguing negative result:
there are infinitely many uncomputable AGM contraction functions in such logics.
Drastically, we also show that the current \emph{de facto} standard strategies to control computability, which rely on restricting the space of epistemic states, fail: uncomputability remains in all non-finitary cases.
Motivated by this disruptive result, we propose new approaches to controlling computability beyond the finitary realm.
Using Linear Temporal Logic (LTL) as a case study, we identify an infinite class of fully-rational AGM contraction functions that are computable by design.
We use Büchi automata to construct such functions, and to represent and reason about LTL beliefs.
\end{abstract}

\section{Introduction}
\label{sec:intro}

Evolving a knowledge base is a crucial problem that has been intensively investigated in several research areas such as in \textit{ontology evolution}, \textit{ontology repair}, \textit{data integration}, and \textit{inconsistency handling}. 
The field of \textit{belief change} \citep{alchourron:partial-meet,gardenfors:flux} investigates this problem from the lense of minimal change: removal of information must be minimised, so most of the original beliefs are preserved. 
The area is founded on the AGM paradigm \citep{alchourron:partial-meet}, which prescribes rationality postulates of minimal change and defines classes of operations that abide by such postulates. 
The removal of  obsolete information is investigated under the name of \textit{contraction}.
Contraction is central, as it underpins most of other kinds of operations and is the core for understanding minimal change. 
For example, to accommodate a new piece of information $\alpha$, one must first remove the potential conflicts with $\alpha$ and then incorporate~$\alpha$. 
The key aspect here is the removal of conflicting information, that is, contraction. 
Minimal change can, therefore, be understood from the lense of contraction itself. %
In this paper, we investigate the \emph{computational aspects} of contraction in non-classical logics.

Although originally developed for classical logics, such as classical propositional logic and first order logic, 
significant efforts have been expended to extend AGM to more expressive non-classical logics used in knowledge representation and reasoning, such as  
\emph{Horn logics} \citep{DelgrandeP15, DelgrandeW10, BoothMVW14}, \emph{para-consistent logics} \citep{da1998belief}, \emph{description logics} \citep{MW08, wassermann:relevance-recovery, flouris:thesis}, and \emph{non-compact logics} \citep{jandson:towards-contraction}.

Despite all these efforts, computational aspects of AGM belief change have received little attention.
The few works on this topic are confined to classical propositional logics and the sub-classical case of Horn logics \Citep{Nebel1998,EiterG92,schwind:compute-iterated}. %
As the majority of the logics in knowledge representation are non-classical, 
for belief change to be properly handled, it is paramount that its computational aspects are investigated in such logics.  
In this paper, we consider a central question: %

\begin{description}
	\item[Computability / Effectiveness:]
	Given a belief change operator $\circ$,
	does there exist a Turing Machine that computes~$\circ$, and stops on all inputs?
\end{description}

This question is trivially answered in the affirmative for the classical finitary case,
that is, when the underlying logic can only distinguish finitely many equivalence classes of formulae, as is the case of classical propositional logic and propositional Horn logic. 
For the non-finitary case,  however, this question is much harder to answer.
We provide, in this paper, a severe and disruptive answer: \textit{AGM contraction suffers from uncomputability, in all non-finitary logics.}
 
The \emph{de facto} standard strategy to control computability rests on limiting ``what can be expressed'', that is, limiting 
the space of epistemic states, in favour of tractability.  %
For instance, families of description logics~\citep{DLBook} have been constructed by %
depriving the object language of the logic of certain connectives,
in favor of taming time and space complexity of some reasoning problems. %

We show that, for AGM contraction, uncomputability is inherent to non-finitary logics and therefore, this strategy of limiting epistemic states has no effect in securing computability. 
This highlights the need for a shift in perspective towards handling computability, which entails devising a \emph{novel machinery to attain computability within AGM}. 
For this, it is paramount to identify how, and under which conditions, one can construct families of computable AGM contraction functions. 
Towards this direction, we examine \emph{Linear Temporal Logic}~\citep{pnueli:ltl}, LTL for short.
LTL is a very expressive logic used in a plethora of applications in Computer Science and AI. %
For example, LTL has been used for specification and verification of software and hardware systems \citep{clarke:model-checking}, in business process models such as DECLARE \citep{AalstPS09}, %
in planning and reasoning about actions \citep{cerrito1998bounded, GiacomoV99},
and extending Description Logics with temporal  knowledge \citep{Gutierrez-Basulto16, Gutierrez-Basulto15}.
We devise a novel machinery for accommodating computability of AGM contraction in LTL. 
We explore \emph{B\"u{chi} automata}~\citep{buchi:automata} as a structure to support knowledge representation and reasoning in LTL, and construct contraction operators upon such automata. %
Our results pave the way for achieving computability of AGM in more general logics used in knowledge representation.
In particular for LTL,
this opens the door to practical applications, for instance in the repair of unrealizable specifications or the repair of incorrect systems.

\textbf{Roadmap:}
In \cref{sec:ltl}, we review basic concepts regarding logics, including LTL and B\"uchi automata.
We briefly review AGM contraction in \cref{sec:contraction}.
\Cref{sec:representation} discusses the question of finite representation for epistemic states,
and presents our first contribution, namely, we introduce a general notion to capture all forms of finite representations, and show a negative result:
for a wide class of so-called \emph{compendious} logics,
not all epistemic states can be represented finitely.
In \cref{sec:buchi-reason}, we present an expressive method of finite representation for LTL based on B\"uchi automata.
In \cref{sec:uncomputability}, we establish our second negative result, for all compendious logics: uncomputability of contraction is inevitable in the non-finitary case.
Towards attaining computability, in \cref{sec:effective-approach}, we identify a large class of computable contraction functions on LTL theories represented via B\"uchi automata.
Computability stems from the fact that the underlying epistemic preference relations are represented as a special kind of automata: \emph{B\"uchi-Mealy automata}.
\Cref{sec:conclusion} discusses the impact of our results and provides an outlook on future work.

Detailed proofs of our results can be found in%
\ifextended\ the appendix.
\else~\cite{kr25:arxiv}.
\fi

\section{Logics and Automata}
\label{sec:ltl}

We review a general notion of logics
that will be used throughout the paper.
We use $\powerset{X}$ to denote the power set of a set $X$.
A \emph{logic} is a pair $\LL = (\Fm, \CnFun)$
comprising a countable%
\footnote{A set $X$ is countable if there is an injection from $X$ to the natural numbers.} set of \emph{formulae} $\Fm$,
and a \emph{consequence operator} ${\CnFun : \powerset{\Fm} \to \powerset{\Fm}}$ that maps each set of formulae to the conclusions entailed from it.
We sometimes write $\Fm_\LL$ and $\CnFun_\LL$ for brevity.

We consider logics that are \emph{Tarskian}, that is, logics whose consequence operator $\CnFun$ is monotone (if $X_1 \subseteq X_2$ then $\Cn{X_1} \subseteq \Cn{X_2}$), extensive ($X\subseteq \Cn{X}$) and idempotent ($\Cn{\Cn{X}} = \Cn{X}$).
We say that two formulae $\varphi,\psi\in\Fm$ are logically equivalent, denoted $\varphi\equiv\psi$, if $\Cn{\varphi} = \Cn{\psi}$.
$\Cn{\emptyset}$ is the set of all tautologies.
A \emph{theory} of $\LL$ is a set of formulae $\kb$ such that $\Cn{\kb} = \kb$.
The expansion of a theory $\kb$ by a formula $\varphi$ is the theory $\kb + \varphi := \Cn{\kb \cup \{\varphi\}}$.
Let $\theories[\LL]$ denote the set of all theories of $\LL$.
If $\theories[\LL]$ is finite, we say that $\LL$ is \emph{finitary};
otherwise, $\LL$ is \emph{non-finitary}.
Equivalently, $\LL$ is finitary if $\LL$ has only finitely many formulae up to logical equivalence.

A theory $\kb$ is \emph{consistent} if $\kb\neq \Fm$,
and it is \emph{complete} if for all formulae $\varphi \notin \kb$, we have $\kb + \varphi = \Fm$.
The set of all complete consistent theories of $\LL$ is denoted as $\CCT[\LL]$.
The set of all CCTs that do not contain $\varphi$ is given by $\compl{\varphi}$.

A logic $\LL$ is \emph{Boolean} if $\Fm_\LL$ is closed under the classical boolean operators and they are interpreted as usual.
In particular, for a logic to be Boolean,
we require every theory $\kb \in \theories[\LL]$ to coincide with the intersection of all the CCTs containing $\kb$, that is, $\kb = \bigcap \{\, \kb' \in \CCT[\LL] \mid \kb \subseteq \kb' \,\}$.

\begin{toappendix}
$\LL$ is Boolean if for every $\varphi,\psi\in \Fm$ there exist formulae $\lnot\varphi$ resp.\ $\varphi\lor\psi$ such that:
\begin{description}
  \item[\textbf{($\lnot_T$)}] $\Cn{\{\varphi\}} \cap \Cn{\{\lnot\varphi\}} = \Cn{\emptyset}$
  \item[\textbf{($\lnot_I$)}] $\Cn{\{\varphi, \lnot\varphi\}} = \Fm$
  \item[\textbf{($\lor_I$)}]  If $\varphi\in\Cn{X}$ then $(\varphi\lor\psi) \in \Cn{X}$
  \item[\textbf{($\lor_E$)}]  If $\alpha\in\Cn{X\cup\{\varphi\}}$ and $\alpha\in\Cn{X\cup\{\psi\}}$,
    then $\alpha\in\Cn{X\cup\{\varphi\lor\psi\}}$
\end{description}
for all formulae $\alpha,\varphi,\psi$ and all sets of formulae $X$;
and we have $\kb = \bigcap \{\, \kb' \in \CCT[\LL] \mid \kb \subseteq \kb' \,\}$ for every theory $\kb$.
\end{toappendix}

We omit subscripts whenever the meaning is clear. 
Given a binary relation $<$ on some domain $D$, the maximal elements of a set $X\subseteq D$ w.rt.\ the relation $<$ are given by

${\max}_{<}(X) := \{\,x\in X\mid \text{there is no } y\in X \text{ s.t. } x < y \,\}$.

\subsection{Linear Temporal Logic}

We recall the definition of \emph{linear temporal logic} \citep{pnueli:ltl}, LTL for short.
For the remainder of the paper, we fix a finite, nonempty set $\AP$ of atomic propositions.
\begin{definition}[LTL Formulae]
Let $p$ range over $\AP$.
\emph{The formulae of LTL} are generated by the following grammar:
\[
  \varphi \mathrel{::=} \bot \ |\ p \ |\ \lnot \varphi \ |\ \varphi \lor \varphi \ |\ \Next\varphi \ |\ \varphi \Until \varphi
\]%
$\Fm_\LTL$ denotes the set of all LTL formulae.
\end{definition}

In LTL, time is interpreted as a linear timeline that unfolds infinitely into the future.
The operator ${\Next\!\!}$ states that a formula holds in the \emph{next} time step,
while $\varphi \Until \psi$ means that $\varphi$ holds \emph{until} $\psi$ holds (and $\psi$ does eventually hold).
We define the usual abbreviations for boolean operations ($\top$, $\land$, $\limplies$),
as well as the temporal operators $\Finally \varphi := \top \Until \varphi$ (\emph{finally}, at some point in the future), $\Globally \varphi := \lnot\Finally\lnot\varphi$ (\emph{globally}, at all points in the future),
and $\Next[k]\varphi$ for repeated application of~${\Next}\!$, where $k\in\N$.

Formally, timelines are modelled as \emph{traces}.
A trace is an infinite sequence $\pi = a_0 a_1 \cdots$,
where each $a_i \in \powerset{\AP}$ is the set of atomic propositions that hold at time step $i$.
The infinite suffix of $\pi$ starting at time step $i$ is denoted by $\pi^i = a_i a_{i+i}\cdots$.
The set of all traces is denoted by $\powerset{\AP}^\omega$.

The semantics of LTL is defined in terms of Kripke structures~\citep{clarke:model-checking},
which describe possible traces.
\begin{definition}[Kripke Structure]
  A \emph{Kripke structure} is a tuple $M = (S, I, T, \lambda)$ where
  $S$ is a finite set of states;
  $I \subseteq S$ is a non-empty set of initial states;
  $T \subseteq S \times S$ is a left-total transition relation, i.e., for all $s \in S$ there exists $s' \in S$ such that $(s,s')\in T$;
  and $\lambda : S \to \powerset{\AP}$ labels states with sets of atomic propositions.
\end{definition}

A trace of a Kripke structure $M$ is a sequence
${\pi = \lambda(s_0)\lambda(s_1)\lambda(s_2)\cdots}$ with
$s_0 \in I$, and for all $i \geq 0$,
$s_i \in S$ and $(s_i, s_{i+1})\in T$.
The set of all traces of a Kripke structure $M$ is given by $\Traces{M}$.
\Cref{fig:kripke} shows an example of a Kripke structure, in graphical notation.

\begin{figure}
  \centering
  \begin{tikzpicture}[node distance=2cm,thick]
    \node[st,rectangle,label={[below=0.4,font=\scriptsize]$\{p\}$},label distance=1] (l0) {$s_0$};
    \node[st,rectangle,left of=l0,label={[below=0.4,font=\scriptsize]$\{\}$},label distance=1] (l1) {$s_1$};
    \node[st,rectangle,right of=l0,label={[below=0.4,font=\scriptsize]$\{p\}$},label distance=1] (l2) {$s_2$};
    \draw[<-] (l0) -- ++(up:0.5);
    \draw[->] (l0) edge[bend left=15] (l1);
    \draw[->] (l0) edge[bend right=15] (l2);
    \draw[->] (l1) edge[bend left=15] (l0);
    \draw[->] (l2) edge[bend right=15] (l0);
  \end{tikzpicture}
  \caption{A Kripke structure on $\AP=\{p\}$, with an initial state~$s_0$. The labels $\lambda(s_i)$ are shown below each state $s_i$.}
  \label{fig:kripke}
\end{figure}
The satisfaction relation between Kripke structures and LTL formulae is defined in terms of the satisfaction between the Kripke structure's traces and LTL formulae.
\goodbreak

\begin{definition}[Satisfaction]\label{def:sattraces}
  The \emph{satisfaction relation} is the least relation ${\models} \subseteq \powerset{\AP}^\omega \times \Fm_\LTL$ between traces and LTL formulae such that, for all $\pi = a_0a_1\cdots \in \powerset{\AP}^\omega$:  %
  \newcommand\myiff{\mbox{iff }}
  \begin{align*}
  	\begin{array}{lcl}
  		\pi \not\models \bot\\
  		\pi \models p &\myiff & p\in a_0\\
  		\pi \models \lnot\varphi &\myiff & \pi \not\models \varphi\\
  		\pi \models \varphi_1 \lor \varphi_2 &\myiff & \pi \models \varphi_1 \text{ or } \pi\models\varphi_2\\
  		\pi \models \Next\varphi &\myiff & \pi^1 \models \varphi\\
  		\pi \models \varphi_1 \Until \varphi_2 &\myiff &\text{ there exists } i \geq 0 \text{ s.t. } \pi^i \models \varphi_2\\
  		      && \text{ and for all } j < i, \pi^j \models \varphi_1
  	\end{array}
  \end{align*}
\end{definition}
A Kripke structure $M$ satisfies a formula $\varphi$, denoted ${M \models \varphi}$, iff all traces of $M$ satisfy $\varphi$.
$M$ satisfies a set $X$ of formulae, $M \models X$, iff $M \models \varphi$ for all $\varphi \in X$.
\begin{example}
  \label{ex:ltl-swim}
  Let the atomic proposition $p$ denote \emph{``Mauricio swims''},
  and let each time step represent one day.
  The LTL formula $\Globally\Finally p$ means \emph{``Mauricio swims infinitely often''} (it always holds that he eventually swims again),
  and is satisfied by the Kripke structure in \cref{fig:kripke}.
  Conversely, the LTL formula $\Globally p$ means \emph{``Mauricio swims every day''}.
  This formula is \emph{not} satisfied by the Kripke structure in \cref{fig:kripke}.

  We return to this example throughout the paper.
\end{example}

The consequence operator $\CnFun_\LTL$ is defined from the satisfaction relation. %
\begin{definition}[Consequence Operator]
  The \emph{consequence operator} $\CnFun_\LTL$ maps each set $X$ of LTL formulae
  to the set of all formulae $\psi$,
  such that for all Kripke structures $M$,
  if~$M \models X$ then also $M \models \psi$.
\end{definition}
\begin{observationrep}
  \label{obs:tarskian-boolean}
  LTL is Tarskian and Boolean.
\end{observationrep}
\begin{proof}
  It is easy to show that LTL is Tarskian.
  Regarding the Boolean operators, the disjunction is straightforward.
  The only interesting aspect is negation:
  \begin{description}
  \item[($\lnot_T$)] Let $\psi\in\Cn{\varphi} \cap \Cn{\lnot\varphi}$, and let $M$ be a Kripke structure.
    Assume $M\not\models\psi$. Then there exists an ultimately periodic trace $\pi$ of $M$ such that $\pi\not\models\psi$.
    But then the Kripke structure $M_\pi$ with $\Traces{M_\pi} = \{\pi\}$ either satisfies $\varphi$ or $\lnot\varphi$,
    and in either case it follows that $M_\pi \models \psi$. Thus we have a contradiction, and it must indeed be the case that every Kripke structure $M$ satisfies $\psi$.
    Hence $\psi \in \Cn{\emptyset}$.
  \item[($\lnot_I$)] Let $\varphi\in\Fm_\LTL$. Since there is no Kripke structure such that $M\models\varphi$ and $M\models\lnot\varphi$,
    we can conclude that all such models satisfy $\psi$.
    Hence $\psi\in\Cn{\{\varphi,\lnot\varphi\}}$.
  \end{description}
\end{proof}
\begin{toappendix}
Note that our notion of negation, in particular $(\lnot_T)$, is weaker than
requiring $M\models\varphi$ or $M\models\lnot\varphi$ for all formulae $\varphi$,
a property not satisfied by LTL.
\end{toappendix}

\subsection{B\"uchi Automata}
\label{sec:buchi-automata}

B\"uchi automata are finite automata widely used in formal specification and verification of systems, especially in LTL model checking %
\citep{clarke:model-checking}. B\"uchi automata have also been used in planning to synthesise plans when goals are  in LTL \citep{GiacomoV99,PatriziLGG11}.  %
\begin{definition}[B\"uchi Automata]%
  A \emph{B\"uchi automaton} is a tuple $A = (Q, \Sigma, \Delta, Q_0, R)$,
  consisting of
  a finite set of states~$Q$;
  a finite, nonempty alphabet $\Sigma$ (whose elements are called \emph{letters});
  a transition relation $\Delta \subseteq Q \times \Sigma \times Q$;
  a set of initial states $Q_0 \subseteq Q$;
  and a set of \emph{recurrence states} $R \subseteq Q$.
\end{definition}

A B\"uchi automaton accepts an infinite word over a finite alphabet $\Sigma$
if the automaton visits a recurrence state infinitely often while reading the word.
\begin{figure}[t]
	\begin{minipage}[t]{0.5\columnwidth}
		\scriptsize
		\textbf{B\"uchi automaton $A_\kb$:}\\[1em]
		\begin{tikzpicture}[background rectangle/.style={fill=gray!15,rounded corners}, show background rectangle]
			\node[st] (q0) {$q_0$};
			\node[st,accepting,right=8mm of q0] (q1) {$q_1$};
			\node[st,accepting,right=8mm of q1] (q2) {$q_2$};

			\draw[<-,thick] (q0) -- ++(left:0.7cm);
			\draw[->,thick] (q0) edge[loop above] node[lbl]{$\emptyset$,$\{p\}$} ();
			\draw[->,thick] (q0) -- node[lbl]{$\emptyset$,$\{p\}$} (q1);
			\draw[->,thick] (q1) edge[bend left] node[lbl]{$\{p\}$} (q2);
			\draw[->,thick] (q2) edge[bend left] node[lbl]{$\emptyset,\{p\}$} (q1);
		\end{tikzpicture}
	\end{minipage}%
	\hfill
	\begin{minipage}[t]{0.45\columnwidth}
		\newcommand{\lap}{\{p\}}
		\newcommand{\lemp}{\emptyset}
		\scriptsize
		\textbf{Some Infinite Words from $\lang{A_{\kb}}$:}\\
		\begin{align*}
			\pi_1 &= \lemp\,\lemp\,\lemp\,\lap\ (\lemp\, \lap)^{\omega} \\
			\pi_2 &= \lap\,\lap\,\lemp\ (\lemp\, \lap)^{\omega}\\
			\pi_3 &= \lap\,\lap\,\lemp\ \lap^\omega\\
		\end{align*}
	\end{minipage}
	\caption{A B\"uchi automaton $A_\kb$ over the alphabet $\Sigma=\{\emptyset,\{p\}\}$.
	  Double circles indicate recurrence states. The initial state $q_0$ is marked by an incoming arrow.
	  On the right, some infinite words accepted by $A_\kb$.
	  By contrast, the word $\emptyset^\omega$ is not accepted.
	}
	\label{fig:buchi-example}
\end{figure}%
\Cref{fig:buchi-example} shows an example of a B\"uchi automaton.

Formally, an infinite word is a sequence $a_0 a_1 \ldots$ with $a_i\in \Sigma$ for all~$i$.
For a finite word $\rho = a_0\ldots a_n$, with $n \geq 0$, let $\rho^\omega$~denote the infinite word corresponding to the infinite repetition of $\rho$.
The set of all infinite words is denoted by $\Sigma^\omega$.
  An infinite word $a_0a_1a_2\ldots \in \Sigma^\omega$ is \emph{accepted} by a B\"uchi automaton $A = (Q, \Sigma, \Delta, Q_0, R)$ if
  there exists a sequence $q_0,q_1,q_2,\ldots$ of states $q_i\in Q$ such that
  $q_0\in Q_0$ is an initial state,
  for all~$i$ we have that $(q_i, a_i, q_{i+1}) \in \Delta$
  and there are infinitely many $i\in\N$ with $q_i \in R$.
  The set~$\lang{A}$ of all accepted words is the \emph{language} of $A$.

\begin{toappendix}
\begin{proposition}[\citep{clarke:model-checking}]
  \label{thm:reg-up}%
  If $\lang{A}$ is nonempty, then $A$ accepts an ultimately periodic word.
\end{proposition}
\end{toappendix}

Emptiness of a B\"uchi automaton's language is decidable.
Further, B\"uchi automata for the union, intersection and complement of the languages of given B\"uchi automata can be effectively constructed~\citep{buchi:automata}.
In the remainder of the paper, we specifically use the construction for union, and denote it with the symbol $\unionbuchi$.
\begin{toappendix}
  In addition to the union construction for B\"uchi automata (denoted by $\unionbuchi$),
  in this appendix we also use the intersection construction and denote it with the symbol $\intersectbuchi$.
\end{toappendix}
Unless otherwise noted, we always consider B\"uchi automata over the alphabet $\Sigma = \powerset{\AP}$,
where letters are sets of atomic propositions and infinite words are traces.
The automata-theoretic treatment of LTL is based on the following result:

\begin{proposition}[\cite{clarke:model-checking}] \label{prop:phi_kripke_to_buchi}%
  \label{prop:ltl-kripke-buchi}%
  For each LTL formula $\varphi$ and Kripke structure $M$,
  there exist  B\"uchi automata~$\ltlbuchi{\varphi}$ and~$\kripkebuchi{M}$
  that accept precisely the traces that satisfy~$\varphi$ resp.\ the traces of~$M$, i.e.,
  $\lang{\ltlbuchi{\varphi}} = \{\,\pi\in\powerset{\AP}^\omega \mid \pi\models\varphi\,\}$,
  and $\lang{\kripkebuchi{M}} = \Traces{M}$.
\end{proposition}

\section{AGM Contraction}
\label{sec:contraction}

In the AGM paradigm, %
the epistemic state of an agent is represented as a theory.
A contraction function for a theory $\kb$ is a function ${\dotmin}: \Fm \to \powerset{\Fm}$ that,
given an unwanted piece of information $\varphi$,
outputs a subset of $\kb$ which does not entail $\varphi$.
Contraction functions are subject to the following
rationality postulates %
\citep{gardenfors:flux}: %

\begin{tabbing}
kkkkkk\= hhhhhghjhkhkjkjhjkhjhjjhjhjhjkjk\=
hjhjhhh\kill
\cp[1] \>  $\kb \dotmin \varphi = \Cn{\kb \dotmin \varphi}$ \> (closure) \\
\cp[2]\> $\kb \dotmin \varphi \subseteq \kb$ \> (inclusion) \\
\cp[3]\>  If $\varphi \not \in \kb$, then $\kb \dotmin \varphi = \kb$ \> (vacuity) \\
\cp[4]\> If $\varphi \not \in \Cn{\emptyset}$, then $\varphi \not \in \kb \dotmin \varphi$ \> (success) \\
\cp[5]\> $\kb \subseteq (\kb \dotmin \varphi) + \varphi$ \> (recovery) \\
\cp[6]\> If $\varphi\equiv\psi$, then $\kb \dotmin \varphi = \kb \dotmin \psi$ \>  (extensionality)\\
\cp[7]\> $(\kb \dotmin\varphi)\cap(\kb \dotmin \psi)\subseteq \kb\dotmin{(\varphi\wedge \psi)}$  \> \\
\cp[8]\> If $\varphi\not\in \kb\dotmin{(\varphi\wedge \psi)}$ then
$\kb\dotmin{(\varphi\wedge \psi)}\subseteq \kb\dotmin\varphi$ \>
\end{tabbing}

For a detailed discussion on the rationale of these postulates,  see \citep{alchourron:partial-meet,gardenfors:flux,hansson:belief-dynamics}.
A contraction function that satisfies \cp[1] to \cp[6]
is called a \emph{rational} contraction function.
If a  contraction function satisfies all the eight rationality postulates,  we say that it is \emph{fully rational}.

There are many different constructions for (fully) rational AGM contraction
on classical logics \citep{hansson:belief-dynamics}.
These contraction functions, however, are not suitable for non-classical logics \citep{flouris:thesis}.
To embrace more expressive logics,  \citet{jandson:towards-contraction} have proposed a new class of (fully) rational contraction functions which only assume the underlying logic to be Tarskian and Boolean: %
the Exhaustive Contraction Functions (for rationality) and the Blade Contraction Functions (for full rationality). We briefly review these functions.

\begin{definition}[Choice Functions]
\label{def:choice-fun}

A \emph{choice function} is a map $\delta: \Fm \to \powerset{\CCT}$ taking each formula $\varphi$ to a set of complete consistent theories satisfying the following:

\begin{description}
	\item[(CF1)] $\delta(\varphi) \not = \emptyset$;
	\item[(CF2)] if $\varphi \not \in \Cn{\emptyset}$, then
        $\delta(\varphi) \subseteq \compl{\varphi}$; and
    \item[(CF3)] for all $\varphi,\psi\in\Fm$, if $\varphi \equiv \psi$ then $\delta(\varphi) = \delta(\psi)$.
\end{description}

\end{definition}

A choice function  chooses at least one complete consistent theory,
for each formula $\varphi$ to be contracted \prop{CF1}.
As long as $\varphi$ is not a tautology, the CCTs chosen must not contain the formula $\varphi$ \prop{CF2},
since the goal is to relinquish $\varphi$.
Choice functions must be syntax-insensitive \prop{CF3}.
\begin{definition}[Exhaustive Contraction Functions]
\label{def:ecf}
Let $\delta$ be a choice function. The
\emph{Exhaustive Contraction Function (ECF)} on a theory $\kb$ induced by $\delta$ is the function $ \dotmin[\delta]$ such that
$\kb \dotmin[\delta] \varphi =
	\kb \cap \bigcap \delta(\varphi), \text{ if } \varphi \notin \Cn{\emptyset} \text{ and } \varphi \in \kb$; otherwise,
	$\kb \dotmin[\delta] \varphi =
\kb $.
\end{definition}
Whenever the formula $\varphi$ to be contracted is not a tautology and is in the theory $\kb$, an ECF modifies the current theory by selecting some CCTs and intersecting them with~$\kb$.
On the other hand, if $\varphi$ is either a tautology or is not in the theory $\kb$, then all beliefs are preserved.

\begin{theorem}\citep{jandson:towards-contraction}
\label{thm:ecf-representation}
A contraction function $\dotmin$ is  rational iff it is an ECF.
\end{theorem}

For full rationality, the choice function must be based on an epistemic preference relation ${{<} \subseteq \CCT \times \CCT}$ on the CCTs.
Intuitively, ${C < C'}$ means that $C'$ is at least as plausible as $C$.
The choice function $\delta_<$ picks the most reliable CCTs w.r.t.\ the preference relation: $\delta_<(\varphi) = \max_{<}(\compl{\varphi})$.
Satisfaction of the postulates \cp[7] and \cp[8] depends on two conditions upon the preference relation:
\begin{description}
	\item \prop{Maximal Cut}: $\max_{<}(\compl{\varphi}) \neq \emptyset$, if $\varphi$ is not a tautology;
	\item \prop{Mirroring}	if $C_1 \not < C_2$ and $C_2  \not < C_1$ but $C_1 < c_3$ then  $C_2 < C_3$
\end{description}

The condition \prop{Maximal Cut} guarantees that for every non-tautological formula, at least one CCT will be chosen for the contraction, ensuring success. As for \prop{Mirroring}, it  imposes that every pair of uncomparable CCTs, $C_1$ and $C_2$, must mimic each other's preferences, that is, a CCT~$C_3$ that is at least as preferable as~$C_1$ must be at least as preferable as~$C_2$.
See \citep{jandson:towards-contraction} for a deep discussion on this property. 
An ECF whose choice function is based on a binary relation satisfying  \prop{Maximal Cut} and \prop{Mirroring} 
is called a Blade Contraction Function.
They are characterised by all rationality postulates.  %

\begin{theorem}\citep{jandson:towards-contraction}
	\label{thm:bcf-representation}
	A contraction function is fully  rational iff it is a Blade Contraction Function. %
\end{theorem}

\section{Finite Representation and its Limits}
\label{sec:representation}

In the AGM paradigm, the epistemic states of an agent are represented as theories
which are in general infinite.
However, according to \citet{Hansson2012,hansson2017descriptor}, the epistemic states of rational agents should have a finite representation.
This is motivated from the perspective that epistemic states should resemble the cognitive states of human minds,
and Hansson argues that as ``finite beings'', humans cannot sustain epistemic states that do not have a finite representation.
Further, finite representation is crucial from a computational perspective,
to represent epistemic states in a computer.
We introduce a general notion of finite representation,
and show that in non-finitary logics, there is no method of finite representation that captures all epistemic states.

	 Different strategies of finite representation have been used such as (i) finite bases \citep{Nebel90:Reasoning,Dalal88,dix94:beliefRevision}, and %
	 (ii) finite sets of models
	 \citep{DEL_book,DEL_intro}. %
		In the former strategy, each finite set $X$ of formulae, called a \emph{finite base}, represents the theory $\Cn{X}$.
	In the latter strategy, %
	models are used to represent an epistemic state. Precisely, each finite set $X$ of models  represents the theory of all formulae satisfied by all models in $X$, that is, the theory $\{  \varphi \in \Fm_{\LL} \mid M \models \varphi, \mbox{ for all } M \in X \}$.
The expressiveness of finite bases and finite sets of models
are, in general (depending on the logic), incomparable, that is, some theories expressible in one method cannot be expressed in the other method and vice versa.
For instance, the information that \emph{``Mauricio swims every two days''} cannot be expressed via a finite base in LTL \citep{wolper:more-expressive}, although it can be expressed via a single Kripke structure (shown in \cref{fig:kripke}, where $p$~again stands for \emph{``Mauricio swims''}, as in~\cref{ex:ltl-swim}).
On the other hand, \emph{``Mauricio will swim eventually''} is expressible as a single LTL formula ($\Finally p$),
but cannot be expressed via a finite set of models.
\begin{toappendix}
  \begin{proposition}
    The theory $\Cn{\Finally p}$ is not expressible via a finite set of models.
  \end{proposition}
  \begin{proof}
    Suppose there was a finite set of models, i.e., Kripke structures $\{M_1,\ldots, M_n\}$
    such that \[\{\,\varphi\in\Fm_\LTL\mid M_i\models\varphi\text{, for } i=1,\ldots,n\,\} = \Cn{\Finally p}\]
    Each Kripke structure $M_i$ has some finite number of states $m_i$.
    Clearly, we must have $M_i \models \Finally p$.
    It follows that for every trace of $M_i$, $p$ must hold at least once within the first $m_i$ time steps:
    otherwise, there must be a cycle that can be reached and traversed without encountering an occurrence of $p$.
    If this were the case, there would also be an infinite trace corresponding to infinite repetition of this cycle, where $p$ never holds;
    this would contradict $M_i \models \Finally p$.

    Let now $m$ be the maximum over all $m_i$ for $i=1,\ldots,n$.
    Then each of the models $M_i$ satisfies $\bigvee_{k=0}^{m} \Next[k] p$.
    Thus, $\bigvee_{k=0}^{m} \Next[k] p$ is in the theory represented by the models $M_1,\ldots, M_n$.
    But this formula is not in $\Cn{\Finally p}$, so we arrive at a contradiction.
  \end{proof}
\end{toappendix}

Given the incomparable expressiveness of these two well-established strategies of finite representations,
it is not clear whether in general, and specifically in non-finitary logics,
there exists a method capable of finitely representing all theories, therefore capturing the whole expressiveness of the logic.
Towards answering this question, we provide a broad definition to conceptualise finite representation.

A finite representation for a theory can been seen as a finite word,
i.e., a \emph{code}, from a fixed finite alphabet $\Sigma_\textnormal{C}$.
For example, the codes $c_1 := \texttt{\string{a, b\string}}$ and $c_2 := \texttt{\string{a, a$\to$b\string}}$ are finite words in the language of set theory, and both represent the theory $\CnFun(\{a \land b\})$.
The set of all codes, i.e., of all finite words over $\Sigma_\textnormal{C}$, is denoted by $\Sigma_\textnormal{C}^*$.
In this sense, a method of finite representation is a mapping $f$ from codes in $\Sigma_\textnormal{C}^*$ to theories.
The pair $(\Sigma_\textnormal{C},f)$ is called an \emph{encoding}.

\begin{definition}[Encoding]
	An \emph{encoding} $(\Sigma_\textnormal{C}, f)$ comprises a finite alphabet $\Sigma_\textnormal{C}$
	and a partial function $f : \Sigma_\textnormal{C}^* \rightharpoonup \theories[\LL]$.
\end{definition}
Given an encoding $(\Sigma_\textnormal{C},f)$,
a word $w\in\Sigma_\textnormal{C}^*$ represents a theory $\kb$, if $f(w)$ is defined and $f(w) = \kb$.
Observe that a theory might have more than one code, whereas for others there might not exist a code.
For instance, in the example above for finite bases, the codes $c_1$ and $c_2$ represent the same theory.
On the other hand, recall that the LTL theory corresponding to \emph{``Mauricio swims every two days''} cannot be expressed in the finite base encoding.
Furthermore, the function $f$ is partial, because not all codes in $\Sigma_\textnormal{C}^*$ are meaningful.
For instance, for the finite base encoding, the code~\texttt{\string{\string{\string}\string}} cannot be interpreted as a finite base.

We are interested in logics which are AGM compliant, that is, logics in which rational contraction functions exist. %
Unfortunately, it is still an open problem %
how to construct AGM contraction functions in all such logics. The most general constructive apparatus up to date, as discussed in \Cref{sec:contraction}, are the Exhaustive Contraction functions proposed by Ribeiro et al. (2018) which assume only few conditions on the logic. %
Additionally, we focus on non-finitary logics, as the finitary case is trivial. We call such logics \textit{compendious}.

\begin{definition}[Compendious Logics]
	\label{def:compendious}
	A logic $\LL$ %
	is \emph{compendious}
	if $\LL$ is Tarskian, Boolean, non-finitary
	and satisfies:
	\begin{description}
		\item[(Discerning)]%
		For all sets $X, Y \subseteq \CCT[\LL]$, we have that $\bigcap X = \bigcap Y$ implies $X = Y$.
	\end{description}
\end{definition}
Compendiousness amounts to expressivity in multiple dimensions.
Compendious logics can express infinitely many distinct sentences (non-finitary),
distinguish between a sentence being true or false (classical negation),
and express uncertainty of two or more sentences (disjunction).
The property \prop{Discerning} is related the possible worlds semantics.
In a possible world, the truth values of all sentences are known.
From this perspective, possible worlds correspond to CCTs.
Under the possible worlds semantics, an agent's epistemic state is interpreted as the exact set of all possible worlds in which all the agent's beliefs are true.
If the truth value of a formula $\varphi$ is unknown,
the agent considers some possible worlds where $\varphi$ is true,
as well as possible worlds where $\varphi$ is false.
Hence, more possible worlds indicate strictly less information.
Equivalently, different sets of possible worlds represent different epistemic states.
This is exactly what \prop{Discerning} conceptualises.
\begin{example}
	Yara and Yasmin encounter a large flightless bird.
	Yara knows that such birds exist in Africa and South America.
	Hence, Yara considers two possible worlds: the bird is from Africa (it is an ostrich), or the bird is from South America (it is a rhea).
	Yasmin, who lived in Australia, believes the bird is an emu (from Australia), a rhea or an ostrich.
	Since Yara and Yasmin consider different sets of possible worlds, their epistemic states differ.
	Yara believes in the disjunction $\mathit{ostrich} \lor \mathit{rhea}$, Yasmin does not.
	She believes only in the disjunction $\mathit{ostrich} \lor \mathit{rhea} \lor \mathit{emu}$.
	As per \prop{Discerning}, Yara and Yasmin present different epistemic states, due to the difference in the considered possible worlds.
\end{example}

The class of compendious logics is broad and includes several widely used logics. %

\begin{toappendix}
\begin{lemmarep}
	\label{lem:disc-unique-decomp}
	A Tarskian, Boolean logic $\LL$ is \prop{Discerning}
	if and only if
	for every $\kb\in\CCT[\LL]$, there exists a formula $\varphi$ with $\kb=\Cn{\varphi}$.
\end{lemmarep}
\begin{proof}
	Let $\LL$ be a Tarskian, Boolean logic.

	\paragraph{\bfseries``$\Rightarrow$'':}
	Assume that $\LL$ satisfies \prop{Discerning},
	i.e., that $\bigcap X = \bigcap Y$ implies $X = Y$ for all sets $X, Y \subseteq \CCT[\LL]$.
	Let $\kb$ be an arbitrary complete consistent theory.

	Consider the set $X = \CCT\setminus \{\kb\}$.
	By assumption, since $X \neq \CCT$,
	it follows that $\bigcap X \neq \bigcap \CCT = \Cn{\emptyset}$.
	Consequently, there must exist some formula $\alpha \in \bigcap X \setminus \Cn{\emptyset}$.

	Every CCT $\kb'\in X$ contains $\alpha$, and hence by consistency we have that $(\lnot\alpha)\notin \kb'$.
	However, since $\alpha$ is by assumption non-tautological,
	there must exist some CCT that does not contain it.
	The only choice is $\kb$,
	and so we conclude that $\alpha\notin \kb$ and by completeness, $(\lnot\alpha)\in \kb$.
	It follows that
	\[
	\Cn{\lnot\alpha} = \bigcap \{\,\hat{\kb}\in\CCT \mid \Cn{\lnot\alpha} \subseteq \hat{\kb}\,\} = \bigcap\{\kb\} = \kb
	\]
	Thus we have shown for an arbitrary CCT $\kb$, that there exists a formula $\varphi$ (namely, $\varphi :\equiv \lnot\alpha$) such that $\kb = \Cn{\varphi}$.

	\paragraph{\bfseries``$\Leftarrow$'':}
	Assume that for every $\kb\in\CCT[\LL]$, there exists a formula $\varphi$ such that $\kb=\Cn{\varphi}$.
	To show that $\LL$ satisfies \prop{Discerning},
	we proceed by contraposition.
	To this end, let $X, Y \subseteq \CCT[\LL]$ such that $X\neq Y$.
	Wlog.\ there exists some $\kb\in X\setminus Y$.
	By assumption, there exists some formula $\varphi \in \Fm_\LL$ such that $\kb = \Cn{\varphi}$.

	Let $\kb'$ be any CCT other than $\kb$.
	If it were the case that $\varphi \in \kb'$,
	it would follow that $\kb \subseteq \kb'$.
	But this is a contradiction, as any strict superset of $\kb$ must be inconsistent.

	Thus for any CCT $\kb'$ other than $\kb$,
	it holds that $\varphi \notin \kb'$ and hence $(\lnot\varphi)\in\kb'$.
	It follows that the formula $\lnot\varphi$ is in the intersection $\bigcap Y$ (since $\kb\notin Y$)
	but not in $\bigcap X$ (since $\kb \in X$).
	We conclude that $\bigcap X \neq \bigcap Y$.
\end{proof}
\end{toappendix}

\begin{theoremrep}\label{prop:compendious_logics_examples}
    The logics LTL, CTL, CTL*, $\mu$-calculus and \emph{monadic second-order logic} (MSO) are compendious.
\end{theoremrep}
\begin{proof}
    We refer to \citep{clarke:model-checking} for the definition of syntax and semantics of CTL, CTL* and the $\mu$-calculus,
    and to \citep{buchi:automata} for MSO (there called SC).
    From these definitions, it is easy to see that these logics are Tarskian and Boolean.
    To show that they are non-finitary, it suffices to find an infinite set of pairwise non-equivalent formulae.
    For the case of LTL, such a set is for instance given for instance by $\{p,\Next p, \Next[2] p,\ldots\}$.
	It remains to show that the logics satisfy \prop{Discerning}.

    We begin by proving this for CTL. The same proof also applies directly to CTL* and $\mu$-calculus (noting that CTL can be embedded in both these logics).
	\citet{browne:characterize-kripke-ctl} show that CTL formulae can characterize Kripke structures up to bisimilarity.
	More precisely, for every Kripke structure $M$, there exists a CTL formula $\varphi_M$ such that $\varphi_M$ is satisfied precisely by Kripke structures that are bisimilar to $M$.
	They also show that bisimilar Kripke structures in general satisfy the same CTL formulae.
	From these results, it follows that every CCT of CTL has a finite base, and thus by \cref{lem:disc-unique-decomp}, \prop{Discerning} follows.
	We first show that every Kripke structure with a single initial state induces a CCT, and conversely, that every CCT is induced by a Kripke structure with a single state.
	The first part is trivial, as it follows directly from the semantics of negation (in CTL) that for every formula $\varphi$ and every Kripke structure $M$ with a single initial state,
	we have that $M \models \varphi$ or $M \models \lnot \varphi$.
	Thus the set of formulae satisfied by $M$ is a complete consistent theory.
	Let now $\kb$ be a CCT.
	Since $\kb$ is consistent, it is satisfied by some Kripke structure $M$.
	Wlog.\ we assume that $M$ has only a single initial state: If not, we make all but one state non-initial; preserving satisfaction of $\kb$.
	Then, by the result of \citet{browne:characterize-kripke-ctl},
	there exists a CTL formula $\varphi_M$ characterizing $M$ up to bisimilarity.
	Since $\kb$ is complete, either $\varphi_M \in \kb$ or $(\lnot\varphi_M)\in\kb$.
	But the latter would contradict $M \models \kb$, hence we know that $\varphi_M \in \kb$.
	It follows that $\Cn{\varphi_M} \subseteq \kb$, and since both are CCTs, this means $\kb = \Cn{\varphi_M}$.

	It remains to prove \prop{Discerning} for LTL and MSO.
	We also achieve this by showing that every CCT has a finite base.
	For LTL, this is shown in \cref{lem:upcct-id,lem:cct-iso} below.
	The proof for MSO is analogous,
	noting that LTL can be embedded in MSO, and that MSO formulae can (like LTL formulae) be expressed as B\"uchi automata~\citep{buchi:automata}.
\end{proof}

It turns out that  there is no method of finite representation capable of capturing all theories in a compendious logic.

\begin{theoremrep}
	\label{thm:imp-fin-repr}
	No encoding can represent every theory of a compendious logic.
\end{theoremrep}%
\begin{inlineproof}[Proof Sketch]%
	We show that, since compendious logics are Tarskian, Boolean and non-finitary,
	there exist infinitely many CCTs.
	From \prop{Discerning}, it follows that there exist uncountably many theories in the logic.
	However, an encoding can represent only countably many theories.
\end{inlineproof}%
\begin{proof}
  Since compendious logics are non-finitary, they have infinitely many theories.
  As the logic is Boolean and Tarskian, every theory can be described as a (possibly infinite) intersection of CCTs.
  Thus, there must be infinitely many CCTs.
  From \prop{Discerning}, it follows that intersections of different sets of CCTs always yield different theories.
  As the powerset of the infinite set $\CCT$ is uncountable, we conclude that there exist uncountably many theories in the logic.
	However, an encoding can represent only countably many theories.
\end{proof}

As not every theory can be finitely represented, only some subsets of theories can be used to express the epistemic states of an agent.
We call a subset $\exc$ of theories an \emph{excerpt} of the logic.
Each encoding induces an excerpt.

\begin{definition}[Finite Representation]
    The excerpt \emph{induced} by an encoding $(\Sigma_\textnormal{C}, f)$
    is the set $\exc := \mathsf{img}(f)$.
    An excerpt induced by some encoding is called \emph{finitely representable}.
\end{definition}

\section{The B\"uchi Encoding of LTL}
\label{sec:buchi-reason}
The encoding in which epistemic states are expressed
crucially determines the tasks that an agent is able to perform.
The encoding must be expressive enough to capture a non-trivial space of epistemic states.
We present a suitable encoding for epistemic states over LTL and show that it is strictly more expressive than traditional strategies.

LTL is commonly used in model checking and planning.
In both these domains, the primary approach to reason about LTL is based on B\"uchi automata.
Thus, B\"uchi automata are predestined to be the basis for an encoding of epistemic states over LTL.
We define the set of LTL formulae represented by a B\"uchi automaton as follows:
\begin{definition}[Support]
	The \emph{support} of a B\"{u}chi automaton $A$ is the set %
	$
	\support{A} := \{\, \varphi \in \Fm_\LTL \mid \forall \pi\in\lang{A}\,.\, \pi \models \varphi \,\}
	$.
	If $\varphi \in \support{A}$, we say that $A$ \emph{supports} $\varphi$.
\end{definition}
\begin{figure}[t]
	\begin{minipage}[t]{0.5\columnwidth}
		\scriptsize
		\textbf{B\"uchi automaton $A_\kb$:}\\[1em]
		\begin{tikzpicture}[background rectangle/.style={fill=gray!15,rounded corners}, show background rectangle]
			\node[st] (q0) {$q_0$};
			\node[st,accepting,right=8mm of q0] (q1) {$q_1$};
			\node[st,accepting,right=8mm of q1] (q2) {$q_2$};

			\draw[<-,thick] (q0) -- ++(left:0.7cm);
			\draw[->,thick] (q0) edge[loop above] node[lbl]{$\emptyset$,$\{p\}$} ();
			\draw[->,thick] (q0) -- node[lbl]{$\emptyset$,$\{p\}$} (q1);
			\draw[->,thick] (q1) edge[bend left] node[lbl]{$\{p\}$} (q2);
			\draw[->,thick] (q2) edge[bend left] node[lbl]{$\emptyset,\{p\}$} (q1);
		\end{tikzpicture}
	\end{minipage}%
	\hfill
	\begin{minipage}[t]{0.45\columnwidth}
		\scriptsize
		\textbf{Supported Formulae:}
		\begin{align*}
			\Finally p &\in \support{A_\kb}\\
			\Globally\Finally p &\in\support{A_\kb}\\
			\Finally\Globally(p \limplies \Next p \lor \Next[2] p) &\in\support{A_\kb}\\
			\Globally p, \lnot(\Globally p) &\notin \support{A_\kb}
		\end{align*}
	\end{minipage}
	\caption{%
	  A B\"uchi automaton, along with some examples of supported (and not supported) LTL formulae.
	}
	\label{fig:buchi-support}
\end{figure}
\begin{example}[continued from \cref{ex:ltl-swim}]
    \label{ex:support-swim}
    \Cref{fig:buchi-support} shows a B\"uchi automaton (on the left),
    along with three supported formulae (on the right):
    \emph{``Mauricio will swim eventually''},
    \emph{``Mauricio swims infinitely often''},
    and the more convoluted belief that
    \emph{``from some point on, if Mauricio swims on a given day, he will also swim the next day or the day after that''}.
	All accepted traces (i.e., for which a run exists that cycles between states~$q_1$ and~$q_2$) satisfy these formulae.

	The formula $\Globally p$ (\emph{``Mauricio swims every day''}) is not supported.
	While the accepted trace $\{p\}^\omega$ satisfies this formula, other accepted traces, such as $\emptyset\,\{p\}^\omega$, do not.
	Consequently, the negation $\lnot(\Globally p)$ is not supported either.
\end{example}

It remains to show that the support of a B\"uchi automaton is a theory.
We observe an intriguing property of B\"uchi automata:
their support is fully determined by those accepted traces~$\pi$ that have the property of being \emph{ultimately periodic}%
, that is,
$\pi = \rho\,\sigma^\omega$ for some finite sequences $\rho,\sigma$.
Recall from \cref{sec:buchi-automata} that the superscript $^\omega$ denotes infinite repetition of the subsequence $\sigma$.
Ultimately periodic traces are tightly connected to CCTs: %
each CCT is satisfied by exactly one ultimately periodic trace.
Let $\UP$ denote the set of all ultimately periodic traces.
The correspondence between CCTs and ultimately periodic traces is formalised by the function $\upCCTFun: \UP \to \CCT[\LTL]$ such that $\upCCT{\pi} = \{ \varphi \in \Fm_{\LTL} \mid \pi \models \varphi \}$.

\begin{toappendix}
As a basis for our results on LTL, we develop a tight connection between ultimately periodic traces and complete consistent formulae.
We begin by defining formulae that uniquely identify an ultimately periodic trace.

\begin{lemmarep}[Identifying Formulae]
	\label{lem:ident-formulae}
	For every ultimately periodic trace $\pi$,
	there exists an LTL formula $\id{\pi}$ that is satisfied by $\pi$ and not by any other trace.
\end{lemmarep}
\begin{proof}
	Let $\pi=\rho\,\sigma^\omega$ be an ultimately periodic trace,
	where $\rho = a_1 \ldots a_n$ and $\sigma = b_0 \ldots b_m$.
	We define the formula
	\begin{align*}
		\id{\pi} :\equiv &\bigg(\bigwedge_{i=1}^n \Next[i-1] a_i\bigg)
		\land \bigg(\bigwedge_{i=0}^m \Next[n+i] b_i \bigg)\\
		&\null \land \Next[n] \Globally \bigg( \bigwedge_{a \in \Sigma} a \to \Next[m+1] a \bigg)
	\end{align*}
	where a letter $a \in \Sigma = \powerset{\AP}$ abbreviates the formula $\bigwedge_{p\in a}p \land \bigwedge_{p\in \AP\setminus a}\lnot p$.

	In this formula, the first conjunct establishes the (possibly empty) prefix $a_1\ldots a_n$.
	The second conjunct establishes the subsequent (non-empty) sequence $b_0\ldots b_m$.
	And finally, the third conjunct describes the shape of the trace, i.e. that after a prefix of length $n$ it becomes periodic with a period of length $m+1$.
	Any trace that satisfies these constraints is necessarily equal to $\pi$.
\end{proof}

In \cref{sec:buchi-reason}, we define the function $\upCCTFun : \UP \to \CCT[\LTL]$ with $\upCCT{\pi} := \{\,\varphi\in\Fm_\LTL\mid \pi \models \varphi\,\}$.
This function can also be expressed via identifying formulae.

\begin{lemma}
  \label{lem:upcct-id}
  It holds that $\upCCT{\pi} = \Cn{\id{\pi}}$.
\end{lemma}
\begin{proof}
  Let $\varphi\in\upCCT{\pi}$.
  Then $\pi\models\varphi$.
  Any Kripke structure $M$ that satisfies $\id{\pi}$ must have $\Traces{M} = \{\pi\}$,
  and hence $M$ also satisfies $\varphi$.
  Therefore we conclude that $\varphi\in\Cn{\id{\pi}}$.

  Conversely, let $\varphi\in\Cn{\id{\pi}}$.
  Consider a Kripke structure $M_\pi$ with $\Traces{M_\pi}=\{\pi\}$.
  As $\pi$ is ultimately periodic, such a Kripke structure (with a finite number of states) exists.
  Then $M_\pi \models \id{\pi}$, so by assumption also $M\models\varphi$.
  But this implies $\pi\models\varphi$ and thus $\varphi\in\upCCT{\pi}$.
\end{proof}

\begin{lemma}
	\label{lem:cct-fun-cct}
	For every $\pi\in\UP$,
	$\upCCT{\pi}$ is a complete consistent theory.
	Hence, the function $\upCCTFun$ is well-defined.
\end{lemma}
\begin{proof}
	Let $\pi=\rho\,\sigma^\omega$ be an ultimately periodic trace,
	where $\rho = a_1 \ldots a_n$ and $\sigma = b_1 \ldots b_n$.
	From \cref{lem:upcct-id}, it immediately follows that $\upCCT{\pi}$ is a theory.

	To show consistency, we identify a model (i.e., a finite Kripke structure) that satisfies every formula in the theory.
	In particular, we construct a Kripke structure $M_\pi = (S, I, \to, \lambda)$
	with  as follows:
	The set of states is given by $S=\{q_1,\ldots,q_n, p_0, \ldots, p_m\}$ with initial states $I = \{q_0\}$.
	We define $\lambda(q_i) = a_i$ for $i\in\{1,\ldots,n\}$,
	and $\lambda(p_j) = b_j$ for $j\in\{0,\ldots,m\}$.
	Finally, $\to$ is the smallest relation with $q_i \to q_{i+1}$, $q_n \to p_0$, $p_j \to p_{j+1}$ and $p_m \to p_0$ for all $i\in\{1,\ldots,n-1\}$ and $j\in\{0,\ldots,m-1\}$.

	This Kripke structure only has a single trace,
	namely $\Traces{M_\pi} = \{\pi\}$.
	Thus it follows that $M_\pi \models \id{\pi}$,
	and consequently $M_\pi \models \Cn{\id{\pi}} = \upCCT{\pi}$.
	Thereby we have shown that $\upCCT{\pi}$ is consistent.

	It remains to show that $\upCCT{\pi}$ is complete,
	i.e., for every $\varphi \in \Fm_\LTL$,
	we must either have $\varphi \in \Cn{\id{\pi}}$
	or $(\lnot\varphi) \in \Cn{\id{\pi}}$.
	We distinguish two cases:

	\begin{description}
		\item[Case 1: $\pi \models \varphi$.]
		Consider some Kripke structure $M$ such that $M \models \Cn{\id{\pi}}$.
		Then every trace of $M$ must satisfy $\id{\pi}$,
		i.e., it must hold that $\Traces{M} = \{\pi\}$.
		Since $\pi \models \varphi$, it follows that $M \models \varphi$.

		This reasoning applies to any $M$ with $M \models \Cn{\id{\pi}}$,
		and thus we have shown that $\varphi \in \Cn{\id{\pi}}$.
		\item[Case 2: $\pi\models\lnot\varphi$.]
		We show that $(\lnot\varphi) \in \Cn{\id{\pi}}$,
		analogously to the previous case.
	\end{description}

	Since one of these two cases always applies, for any $\varphi$,
	we have shown that $\upCCT{\pi}$ is complete.
\end{proof}

\begin{lemma}
	\label{lem:cct-fun-inj}
	The function $\upCCTFun$ is injective.
\end{lemma}
\begin{proof}
	Let $\pi_1,\pi_2 \in \UP$ be ultimately periodic traces
	such that $\upCCT{\pi_1} = \upCCT{\pi_2}$.
	Since $\pi_2 \models \id{\pi_2}$
	it follows that $\pi_2 \models \varphi$ for any $\varphi \in \Cn{\id{\pi_2}} = \upCCT{\pi_2}$.
	But since $\id{\pi_1} \in \upCCT{\pi_1}$, and the two theories are equal,
	this means that $\pi_2 \models \id{\pi_1}$.
	By \cref{lem:ident-formulae}, we conclude that $\pi_1 = \pi_2$.
	Thus the function $\upCCTFun$ is injective.
\end{proof}

\begin{lemma}
	\label{lem:cct-fun-surj}
	The function $\upCCTFun$ is surjective on $\CCT[\LTL]$.
\end{lemma}
\begin{proof}
	Let $\kb$ be a complete consistent theory.
	Since $\kb$ is consistent, there exists a Kripke structure $M$ such that $M \models \kb$.
	Like any finite Kripke structure, $M$ contains at least one ultimately periodic trace $\pi$.
	We will show that $\kb = \upCCT{\pi}$, by considering each inclusion separately.

	\begin{description}
		\item[$\kb \subseteq \upCCT{\pi}$:] Let $\varphi \in \kb$.
		Since $M \models \kb$ and $\pi\in\Traces{M}$, we know that $\pi \models \varphi$.
		It follows that $\id{\pi} \models \varphi$, and hence $\varphi \in \Cn{\id{\pi}} = \upCCT{\pi}$.
		\item[$\upCCT{\pi} \subseteq \kb$:] Let $\varphi \in \upCCT{\pi}$.
		Then we know that $\pi \models \varphi$ and thus $\pi \not\models \lnot\varphi$.
		It follows that also $M \not\models \lnot\varphi$.
		Since $M \models \kb$, this means that $(\lnot\varphi) \notin \kb$.
		But since $\kb$ is complete, we conclude that $\varphi \in \kb$.
	\end{description}

	Thereby we have shown that any CCT is equal to $\upCCT{\pi}$ for some ultimately periodic trace $\pi$,
	and thus the function $\upCCTFun$ is surjective on $\CCT[\LTL]$.
\end{proof}
\end{toappendix}

\begin{lemmarep}%
	\label{lem:cct-iso}
	The function $\upCCTFun$ is a bijection.
\end{lemmarep}
\begin{proof}
	This follows from \cref{lem:cct-fun-inj,lem:cct-fun-surj}.
\end{proof}

\begin{toappendix}
  We have shown that in compendious logics, every CCT $\kb$ has a finite base, i.e., a formula $\varphi$ with $\kb=\Cn{\varphi}$.
  \Cref{lem:cct-iso,lem:upcct-id} give us a concrete idea of these finite bases for the case of LTL:
  every CCT of LTL is equal to $\Cn{\id{\pi}}$, for some ultimately periodic trace $\pi$.

  Next, we make use of this connection between CCTs and ultimately periodic traces to characterize the support of a B\"uchi automaton.
\end{toappendix}

We combine \cref{lem:cct-iso} with two classical observations~\citep{clarke:model-checking}:
(i) every consistent LTL formula is satisfied by at least one ultimately periodic trace;
and (ii) every B\"uchi automaton with nonempty language accepts some ultimately periodic trace.
We arrive at the following characterization:

\begin{lemmarep}
\label{lem:buchi-support-cct}
The support of a B\"uchi automaton $A$ satisfies
		\begin{equation*}
			\support{A} = \bigcap \{\, \upCCT{\pi} \mid \pi\in\lang{A}\cap\UP \,\}\,.
		\end{equation*}
\end{lemmarep}
\begin{proof}
  Let $\varphi\in\support{A}$.
  Then $\pi\models\varphi$ for each $\pi\in\lang{A}$,
  and in particular, for each $\pi\in\lang{A}\cap\UP$.
  Thus, $\varphi\in\upCCT{\pi}$ for each such ultimately periodic $\pi$,
  and hence $\varphi\in\bigcap \{\, \upCCT{\pi} \mid \pi\in\lang{A}\cap\UP \,\}$.

  For the converse inclusion,
  let $\varphi\in\bigcap \{\, \upCCT{\pi} \mid \pi\in\lang{A}\cap\UP \,\}$.
  Then $\pi\models\varphi$ for each ultimately periodic trace in $\lang{A}$.
  Suppose there was a trace $\pi'$ that was not ultimately periodic, such that $\pi'\not\models \varphi$.
  Then the set $\lang{A}\setminus\lang{\ltlbuchi{\varphi}}$ would be non-empty.
  As the difference of two languages recognized by B\"uchi automata can again be recognized by a B\"uchi automaton,
  and any B\"uchi automaton that recognizes a nonempty language accepts at least one ultimately periodic trace,
  we conclude that there exists an ultimately periodic trace in $\lang{A}$ that does not satisfy $\varphi$.
  This is however a contraction.
  Hence, our assumption was incorrect and indeed we have $\pi'\models\varphi$ for all $\pi'\in\lang{A}$.
  We conclude that $\varphi\in\support{A}$.

  We have shown both inclusions, so the equality holds.
\end{proof}

\begin{theoremrep}
	The support of a B\"uchi automaton is a theory.
\end{theoremrep}
\begin{proof}
  This is a direct consequence of \cref{lem:buchi-support-cct},
  as the intersection of (complete consistent) theories is a theory.
\end{proof}

Thus, B\"uchi automata indeed define an encoding.
Every B\"uchi automaton $A$, being a finite structure, can be described in a finite code word $w_A$, which the encoding maps to the theory $\support{A}$.
We call this encoding the \emph{B\"uchi encoding}, denoted $(\alphBuchi, \encBuchi)$, and the induced excerpt the \emph{B\"uchi excerpt} $\excBuchi$.
The B\"uchi excerpt is strictly more expressive than the classical strategies of finite representation discussed in \cref{sec:representation}:

\begin{theorem}%
	\label{thm:buechi-expr}
	Let $\exc_\textnormal{base}$ and $\exc_\textnormal{models}$ denote respectively the excerpts of  finite bases and finite sets of models.
	It holds that $\exc_\textnormal{base} \cup \exc_\textnormal{models} \subsetneq \excBuchi$.
\end{theorem}
\begin{proof}[Proof Sketch]
	The expressiveness of the B\"uchi excerpt follows from \cref{prop:ltl-kripke-buchi}.
	\Cref{fig:buchi-support} shows an automaton whose support can be expressed neither by a finite base nor a finite sets of models.
\end{proof}

\begin{toappendix}
In terms of reasoning, the B\"uchi encoding benefits from the decidability properties of B\"uchi automata.
Many decision problems, most importantly the entailment problem on the B\"uchi encoding, can be reduced to the decidable problem of inclusion between B\"uchi automata.

\begin{theorem}
	\label{thm:buchi-decide}
	The following problems wrt.\ the B\"uchi encoding are decidable,
	where $w,w'\in \alphBuchi^*$, $\varphi\in\Fm_\LTL$, and $M$ a Kripke structure:
	\begin{description}
	\item \emph{entailment:} Given $(w,\varphi)$, decide $\varphi \in \encBuchi(w)$.
	\item \emph{model consistency:} Given $(w, M)$, decide $M \models \encBuchi(w)$.
	\item \emph{inclusion:} Given $(w,w')$, decide $\encBuchi(w) \subseteq \encBuchi(w')$.
	\end{description}
\end{theorem}
\begin{proof}
	The problems can be reduced to automata inclusions.
	Particularly, %
	if $w,w'$ encode B\"uchi automata $A,A'$,
	the problems above correspond to the inclusions
	$\lang{A} \subseteq \lang{\ltlbuchi{\varphi}}$, $\lang{\kripkebuchi{M}} \subseteq \lang{A}$ resp.\ $\lang{A'} \subseteq \lang{A}$.
\end{proof}

Beyond ensuring the decidability of key problems,
an encoding's suitability for reasoning also involves the question
whether modifications of epistemic states can be realized by computations on code words.
In particular in the context of the AGM paradigm,
it is interesting to see if belief change operations can be performed in such a manner.
The B\"uchi encoding also shines in this respect,
since we can employ automata operations to this end.
As a first example, consider the \emph{expansion} of a theory $\kb$ with a formula $\varphi$.
This operation corresponds to an intersection operation on B\"uchi automata, %
as the support of a B\"uchi automaton satisfies $\support{A} + \varphi = \support{A \intersectbuchi \ltlbuchi{\varphi}}$.
The intersection automaton $A\sqcap \ltlbuchi{\varphi}$ can be computed through a product construction.
  \begin{lemma}
    Let $A$ be a B\"uchi automaton, and $\varphi$ and LTL formula.
    Then it holds that $\support{A} + \varphi = \support{A\intersectbuchi\ltlbuchi{\varphi}}$.
  \end{lemma}
  \begin{proof}
    Let $\psi \in \support{A} + \varphi = \Cn{\support{A}\cup\{\varphi\}}$.
    By \cref{lem:buchi-support-cct}, it suffices to show that each ultimately periodic trace $\pi\in\lang{A\intersectbuchi\ltlbuchi{\varphi}}$ satisfies $\psi$.
    To see this, consider a Kripke structure $M_\pi$ with $\Traces{M_\pi} = \{\pi\}$.
    Clearly, $\pi\models\support{A} \cup \{\varphi\}$, and so $M_\pi \models \support{A} \cup \{\varphi\}$.
    This implies that $M_\pi\models\psi$, and hence $\pi\models\psi$.
    As this holds for all ultimately periodic traces $\pi\in\lang{A\intersectbuchi\ltlbuchi{\varphi}}$, we conclude that $\psi\in\support{A\intersectbuchi\ltlbuchi{\varphi}}$.

    For the converse inclusion,
    let $\psi\in\support{A\intersectbuchi\ltlbuchi{\varphi}}$.
    We have to show that any Kripke structure $M$ with $M\models\support{A}\cup\{\varphi\}$ also satisfies $\psi$.
    Suppose this was not the case, i.e., that $M\not\models\psi$.
    Then there exists an ultimately periodic trace $\pi\in\lang{\kripkebuchi{M}} \setminus \lang{\ltlbuchi{\psi}}$.
    As $\pi\notin\lang{A}$, we have that every trace in $\lang{A}$ satisfies $\lnot\id{\pi}$, and hence $(\lnot\id{\pi})\in\support{A}$.
    But this contradicts the fact that $M\models\support{A}$.
    Hence the supposition was wrong, and we have indeed that $M\models\psi$.

    We have shown both inclusions, so the equality holds.
  \end{proof}
\end{toappendix}

\section{The Impossibility of Effective Contraction}
\label{sec:uncomputability}

Assume that the space of epistemic states that an agent can entertain is determined by an excerpt $\exc$.
In this section,  we investigate which properties make an excerpt suitable from the AGM vantage point %
and its computability aspects.
Clearly, not every excerpt is suitable for representing the space of epistemic states.
For example, if a non-tautological formula $\varphi$ appears in each theory of $\exc$, then %
$\varphi$ cannot be contracted.
The chosen excerpt should be expressive enough to describe all relevant epistemic states that an agent
might hold in response to its beliefs in flux. %
Precisely, if an agent is confronted with a piece of information and changes its epistemic state into a new one, then the new epistemic state must be expressible in the excerpt. %
A solution is to require the excerpt to contain at least one rational outcome for each possible contraction.
We say that a contraction ${\dotmin}$ \emph{remains in $\exc$} if $\mathsf{img}(\dotmin) \subseteq \exc$.

\begin{definition}[Accommodation]
	An excerpt $\exc$ \emph{accommodates (fully) rational contraction}
	if for each $\kb\in\exc$ there exists a (fully) rational contraction on $\kb$ that remains in $\exc$.
\end{definition}

Accommodation guarantees that an agent can modify its beliefs rationally, in all possible epistemic states covered by the excerpt.
There is a clear connection between accommodation and AGM compliance~\citep{flouris:thesis}.
While AGM compliance concerns existence of rational contraction operations in every theory of a logic,  accommodation guarantees that the information in each theory within the excerpt can be rationally contracted and that its outcome can yet be expressed within the excerpt. %
Surprisingly, rational accommodation and fully rational accommodation coincide. %

\begin{toappendix}

\begin{definition}
	A theory is $\kb$ is supreme iff $\kb$ is not tautological and
	for all $\alpha \in \kb $, either $\Cn{\alpha} = \kb$ or $\Cn{\alpha} = \Cn{\emptyset}$.
\end{definition}

Observe that by definition, supreme theories always have a finite base.

Let $r$ be a function that ranks each CCT to a negative integer, such that distinct CCTs are ranked to different negative integers. %
Let $<_{r}$ be the induced relation from $r$  that is
$X <_r Y$ iff $r(X) < (Y)$. Note that $<_{r}$ is a strict total order. Also note that it satisfies \prop{Maximal Cut} and due to totality it satisfies \prop{Mirroring}.

Given a theory $\kb$, we define

\begin{align*}
	\kb \mathbin{\circ} \varphi &=
	\begin{cases}
		\kb \cap \bigcap  \min_{<_{r}}(\compl{\varphi}) & \mbox{ if } \varphi \not\equiv \top \mbox{ and } \varphi \in \kb \\
		\kb & \mbox{ otherwise}
	\end{cases}
\end{align*}

\begin{observation}\label{th:circ_full}
	$\kb \mathbin{\circ} \varphi$ is fully AGM rational.
\end{observation}

\begin{proof}
	Observe that by definition, $\mathbin{\circ}$ is a blade contraction function and therefore it is fully AGM rational.
\end{proof}

It remains to show that $\mathbin{\circ}$ remains within the excerpt.

\begin{proposition}\label{prop:circ_rem}
	If $\exc$ accommodates contraction and   $\kb \in \exc$, then
	for all formula $\varphi$, $\kb \mathbin{\circ} \varphi \in \exc$
\end{proposition}
\begin{proof}
	If $\varphi \in \Cn{\emptyset}$ or $\varphi \not \in \kb$, then by definition $\kb \mathbin{\circ} \varphi = \kb$, and by hypothesis, $\kb$ is within the excerpt.
	The proof proceeds for the case that $\varphi \not \equiv \top$ and $\varphi \in \kb$. Let
	$\kb' =  \kb \mathbin{\circ} \varphi$.
	Thus, from the definition of $\mathbin{\circ}$, we have
	$\kb' = \kb \cap  \bigcap \min_{<_\kb}(\compl{\varphi}).$
	As $<_{r}$ is strictly total,  $\min_{<_{r}}(\compl{\varphi})$ is a singleton set $\{M\}$,  which implies that
	$ \kb' = \kb \cap M$.

	Let $X = \CCT \setminus \{M\}$.
	Thus, $\bigcap X$ is a supreme theory, which means that there is some formula $\alpha$ such that
	$\bigcap X = Cn(\alpha)$.
	Therefore, as $M$ is the only counter CCT of $\alpha$, we get that the only solution to contract $\alpha$ is $\kb \cap M$.
	By hypothesis, the excerpt $\exc$ accommodates contraction. Thus, there is some contraction operator $\dotmin$ on $\kb$ such that ${\mathrm{img}}(\dotmin) \subseteq \exc$. Therefore,
	$\kb \dotmin \alpha = \kb \cap M$.
	Thus, as by hypothesis $\dotmin$ remains within the excerpt, we have that $\kb \dotmin \alpha \in \exc$ which implies that
	$\kb' \in \exc$. Therefore,
	as $\kb' = \kb \mathbin{\circ} \varphi$, we have that $\kb \mathbin{\circ} \varphi \in \exc$.
\end{proof}

\end{toappendix}

\begin{propositionrep}
    \label{prop:accom-equiv}
	An excerpt $\exc$ accommodates rational contraction
	iff
	$\exc$ accommodates fully rational contraction.
\end{propositionrep}
\begin{proof}
    The fact that accommodation of fully rational contraction implies accommodation of rational contraction is straightforward.
    The opposite direction follows from \cref{th:circ_full} and \cref{prop:circ_rem}.
\end{proof}

Accommodation is the weakest condition we can impose upon an excerpt to guarantee the existence of AGM rational contractions.
Yet, the existence of contractions does not imply that an agent can \emph{effectively} contract information.
Thus we investigate the question of \emph{computability} of contraction functions.
For this endeavor, the focus on contraction functions that remain in the excerpt is crucial:
both input and output of a computation must be finitely representable.
We thus fix a finitely representable excerpt $\exc$ that accommodates contraction.
As an agent has to reason about its beliefs,
it should be able to decide whether two formulae are logically equivalent.
Hence, we assume that, in the underlying logic, logical equivalence is decidable.

\begin{definition}[AGM Computability]
	Let $\kb$ be a theory in~$\exc$,
	and let ${\dotmin}$ be a contraction function on $\kb$ that remains in $\exc$.
	We say that ${\dotmin}$ is \emph{computable} if
	there exists an encoding $(\Sigma_\textnormal{C}, f)$ that induces $\exc$,
	such that the following problem is computed by a Turing machine:%
	\begin{center}
		\fbox{\begin{minipage}{0.95\columnwidth}
				\begin{description}
					\item[Input:] A formula $\varphi \in \Fm_\LL$.
					\item[Output:] A word $w\in\Sigma_\textnormal{C}^*$ such that $f(w) = \kb \dotmin \varphi$.
				\end{description}
			\end{minipage}
		}
	\end{center}
\end{definition}
\goodbreak

In the classical setting of finitary logics,
computability of AGM contraction is trivial, as
there are only finitely many formulae (up to equivalence),
and only a finite number of theories.
By contrast,
compendious logics have infinitely many formulae (up to equivalence) and consequently infinitely many theories.
In the following, unless otherwise stated, we only consider compendious logics.
In such logics, we distinguish two kinds of theories:
those that contain infinitely many formulae (up to equivalence),
and those that contain only finitely many formulae (up to equivalence).
An excerpt that constrains an agent's epistemic states to the latter case
essentially disposes of the expressive power of the compendious logic,
as in each epistemic state only finitely many sentences can be distinguished.
Therefore, such epistemic states could be expressed in a finitary logic.
As the computability in the finitary case is trivial,
we focus on the more expressive case.

\begin{definition}[Non-Finitary]
	A theory $\kb$ is \emph{non-finitary}
	if $\kb$~contains infinitely many logical equivalence classes of formulae.
\end{definition}

Note that being non-finitary is a very general condition.
Even theories with a finite base can be non-finitary.
For instance, the LTL theory $\Cn{\Globally p}$ %
contains the infinitely many non-equivalent formulae $\{p, \Next p, \Next[2] p, \Next[3] p, \ldots \}$.
\goodbreak

\begin{toappendix}
	Let us see an example of finitary theories in LTL.
	\begin{example}
		For an ultimately periodic trace $\pi\in\UP$, consider the theory $\Cn{\lnot\id{\pi}}$.
		This theory contains only 2 equivalence classes,
		namely the equivalence class of $\lnot\id{\pi}$ and the tautologies.
		Specifically, consider some $\varphi$ with $\lnot\id{\pi} \models \varphi$.
		This means that for all $\pi'\in\UP\setminus\{\pi\}$, we have that $\pi'\models\varphi$.
		Now either $\pi\models\varphi$, and hence $\varphi$ is a tautology;
		or $\pi\not\models\varphi$,
		so
		\[
		\{\pi'\in\UP\mid \pi'\models\varphi\} \subseteq \UP \setminus \{\pi\} = \{\pi'\in\UP \mid \pi'\models\lnot\id{\pi}\}
		\]
		which means $\varphi\models\lnot\id{\pi}$,
		and hence (with the assumption above) $\varphi\equiv\lnot\id{\pi}$.
	\end{example}
	The formulae $\lnot\id{\pi}$ are the weakest non-tautological formulae of LTL;
	their negation $\id{\pi}$ are the bases of CCTs (i.e., they are the strongest consistent formulae of LTL).
	In general, a finitary belief state must be very weak;
	it can only imply finitely many formulae.
	Note that the implied formulae can be weakened further by disjoining them with arbitrary other formulae;
	and still such weakening only results in finitely many different beliefs.
	In other words: A finitary theory is only finitely many beliefs away from the tautological theory; and those finitely many beliefs must be very weak, or they would imply infinitely many consequences.
\end{toappendix}

In the remainder of this section, we establish a strong link between non-finitary theories
and uncomputable contraction functions.
To this end, we introduce the notion of \emph{cleavings}.

\begin{definition}[Cleaving]
	A \emph{cleaving} is an infinite set of formulae $\cleav$ such that for all two distinct $\varphi,\psi\in \cleav$ we have:
	\begin{description}
		\item[(CL1)] $\varphi$ and $\psi$ are not equivalent ($\varphi\not\equiv\psi$); and
		\item[(CL2)] the disjunction $\varphi\lor\psi$ is a tautology.
	\end{description}
\end{definition}

\begin{example}
  Consider the logic of elementary arithmetic over natural numbers.
  The formulae ${x\neq 0}$, ${x\neq 1}$, ${x\neq 2}$, etc.\ form a cleaving:
  they are pairwise non-equivalent, and every disjunction $(x\neq \mathsf{n}) \lor (x\neq \mathsf{m})$,
  equivalently written as $\lnot(x=\mathsf{n}\land x=\mathsf{m})$,
  is a tautology (for constants $\mathsf{n}\neq \mathsf{m}$).
\end{example}

From an algebraic perspective, the formulae in a cleaving behave like a kind of weak complement:
we require that the disjunction $\varphi\lor\psi$ is a tautology, whereas we do not require the conjunction $\varphi\land\psi$ to be inconsistent (as would be the case for the conjunction $\varphi\land\lnot\varphi$).

\begin{toappendix}
	
\subsection{Existence of Infinite Cleavings}
We prove that every non-finitary theory must contain an infinite cleaving.

\begin{definition}
	The decomposition of a theory in terms of CCTs is given by the function
	$$ \decomp{\kb} = \{ X \in \CCT \mid \kb \subseteq X\}.$$
\end{definition}

\begin{lemma}
  A theory is non-finitary iff $\CCT \setminus \decomp{\kb}$ is infinite.
\end{lemma}
\begin{proof}
	We show the two implications separately.
	\begin{description}
		\item[``$\Rightarrow$'':]
		Contrapositively, assume that $\CCT \setminus \decomp{\kb}$ is finite,
		in particular let $\CCT \setminus \decomp{\kb} = \{ C_1,\ldots, C_n\}$.
		We know that in a logic with the above assumptions, every CCT has a finite base.
		In particular, let $C_i = \Cn{\varphi_i}$ for $i=1,\ldots,n$.
		Then every formula in $\kb$ is equivalent to $\bigwedge_{C_i\in X}\lnot\varphi_i$ for some $X\subseteq \{C_1,\ldots,C_n\}$.

		To see this, take some $\alpha\in\kb$.
		Let $X = \{ C_i \mid \alpha\notin C_i \}$.
		Then $\varphi_i \models \lnot\alpha$ for every $C_i\in X$,
		therefore $\alpha \models \lnot\varphi_i$,
		and thus $\alpha \models \bigwedge_{C_i\in X}\lnot\varphi_i$.
		For the reverse entailment, note that every CCT not in $X$ is either one of the remaining $C_i$ or in $\decomp{\kb}$,
		and thus every such CCT contains $\alpha$.
		Therefore, any CCT containing the formula $\bigwedge_{C_i\in X}\lnot\varphi_i$ must also contain $\alpha$.
		In other words, $\bigwedge_{C_i\in X}\lnot\varphi_i \models \alpha$.

		Since every equivalence class of formulae in $\kb$ corresponds to one of the $2^n$ possible choices of $X$,
		we conclude that $\kb$ is finitary.

		\item[``$\Leftarrow$'':]
		Suppose $\CCT \setminus \decomp{\kb}$ is infinite,
		and let $\CCT \setminus \decomp{\kb} = \{C_1,C_2,\ldots\}$ be a duplicate-free enumeration of the set.
		For every $C_i$,
		we know that there exists a finite base $\varphi_i$.
		Consequently, the infinitely many formulae $\lnot\varphi_i$ are all in $\kb$.
		These formulae are pairwise non-equivalent (otherwise we would have $C_i=C_j$).
		Thus we conclude that $\kb$ is non-finitary.
	\end{description}
\end{proof}

\end{toappendix}

\begin{lemmarep}
	\label{lem:nonfin-cleaving}
	Every non-finitary theory contains a cleaving.
\end{lemmarep}
\begin{proof}
	Let $\kb$ be a non-finitary theory.
	By the above lemma, we know that $\CCT \setminus \decomp{\kb}$ is infinite.
	Let $\CCT \setminus \decomp{\kb} = \{C_1,C_2,\ldots\}$ be a duplicate-free enumeration of the set,
	and let $C_i = \Cn{\varphi_i}$ for each $i$.
	We consider the set of formulae $\{\,\lnot\varphi_i \mid C_i \in \CCT \setminus \decomp{\kb}\,\}$.
	This set is infinite, and the formulae are pairwise non-equivalent.
	For every pair $C_i,C_j$ with $i \neq j$,
	the formula $(\lnot\varphi_i)\lor(\lnot\varphi_j)$ is a tautology.
\end{proof}

\begin{example}
  \label{ex:cleaving-ltl}
  Returning to our swimming example for LTL,
  consider the following statement:
  \begin{quotation}%
    \noindent\upshape
    If Mauricio will swim in $n$ days from today, he will swim on at least two days (overall).
  \end{quotation}
  This can be written as the LTL formula $\psi_n$ with
  \[
  	\psi_n :\equiv (\Next[n] p) \to \mathit{twice}(p)\ ,
  \]
  where the LTL formula $\mathit{twice}(p) :\equiv \Finally (p\land\Next\Finally p)$ expresses that Mauricio swims on at least two days.
  The set of formulae $\{\,\psi_n\mid n\in\mathbb{N}\,\}$ is a cleaving in the theory $\support{A_\kb}$ supported by the B\"uchi automaton in \cref{fig:buchi-support}:
  \begin{itemize}
    \item Each formula $\psi_n$ is in the theory.
      As shown in \cref{fig:buchi-support}, the formula $\Globally\Finally p$ (\emph{``Mauricio swims infinitely often''}) is in the theory, and it implies (the conclusion of) each $\psi_n$.
    \item Whenever $n\neq m$, the formulae $\psi_n$ and $\psi_m$ are not equivalent \textbf{\upshape(CL1)}.
    \item Whenever $n\neq m$, the disjunction $\psi_n \lor \psi_m$ is equivalent to
      $(\Next[n] p) \land (\Next[m] p) \to \mathit{twice}(p)$,
      a tautology \textbf{\upshape(CL2)}:
      if Mauricio swims in $n$ days and in $m$ days, he clearly swims on at least two days.
  \end{itemize}
\end{example}

Given a contraction that remains in an excerpt, cleavings provide a way of generating many contractions that also remain in the excerpt.
This works by ranking the formulae in the cleaving
such that each rank has exactly one formula.
We reduce the contraction of a formula $\varphi$ to contracting $\varphi \lor \psi$, where $\psi$ is the lowest ranked formula in the cleaving such that $\varphi\lor\psi$ is non-tautological.
Each new contraction depends on the original choice function and the ranking. %

\begin{definition}[Composition]
	Let $\delta$ be a choice function on a theory $\kb$,
	let $\cleav\subseteq \kb$ be a cleaving,
	and let ${\pi : \N \to \cleav}$ be a permutation of $\cleav$.
	The \emph{composition of $\delta$ and $\pi$} is the function $\delta_{\pi}: \Fm \to \powerset{\CCT}$ such that
	\[
	\delta_\pi(\varphi) := \delta\big(\varphi \lor \minCleav{\pi}{\varphi}\big)\ ,
	\]
	where $\minCleav{\pi}{\varphi} = \pi(i)$, for the least $i\in\N$ such that $\varphi\lor\pi(i)$ is non-tautological,
	or $\minCleav{\pi}{\varphi} =\bot$ if no such $i$ exists.
\end{definition}
\goodbreak

The composition of a choice function $\delta$ with a permutation of a cleaving  preserves rationality.

\begin{lemmarep}\label{lem:comp-choice-rational}
	The composition $\delta_\pi$ of a choice function $\delta$
	and a permutation $\pi$ of a cleaving $\cleav\subseteq\kb$
	is a choice function.
\end{lemmarep}
\begin{proof}
	We show that $\delta_\pi$ satisfies all three conditions of choice functions, for all formulae $\varphi,\psi$:

	\begin{description}
		\item[To show: $\delta_\pi(\varphi) \neq \emptyset$.]
		Since $\delta_\pi(\varphi) = \delta(\varphi\lor{\min}_\pi(\varphi))$,
		and by assumption that $\delta$ is a choice function,
		we have $\delta(\varphi\lor{\min}_\pi(\varphi)) \neq \emptyset$, the result follows.
		\item[To show: If $\varphi\notin \Cn{\emptyset}$, then $\delta_\pi(\varphi) \subseteq \compl{\varphi}$.]
		Suppose that $\varphi\notin\Cn{\emptyset}$.
		We have either ${\min}_\pi(\varphi) = \pi(i)$ for some $i$,
		or ${\min}_\pi(\varphi) = \bot$.
		In the latter case, $\varphi \lor {\min}_\pi(\varphi) \equiv \varphi$,
		and thus by assumption that $\delta$ is a choice function,
		$\delta(\varphi \lor {\min}_\pi(\varphi)) = \delta(\varphi) \subseteq \compl{\varphi}$.

		Let us thus now assume that ${\min}_\pi(\varphi) = \pi(i)$ for some $i$.
		Then $\compl{\varphi\lor\pi(i)} = \compl{\varphi} \cap \compl{\pi(i)} \neq \emptyset$,
		so $\varphi\lor\pi(i)$ is not a tautology.
		Since $\delta$ is a choice function,
		we conclude that $\delta_\pi(\varphi) = \delta(\varphi\lor\pi(i)) \subseteq \compl{\varphi\lor\pi(i)}\subseteq \compl{\varphi}$.
		\item[To show: If $\varphi\equiv\psi$, then $\delta_\pi(\varphi) = \delta_\pi(\psi)$.]
		Suppose $\varphi\equiv\psi$,
		then we have $\compl{\varphi} = \compl{\psi}$,
		and thus ${\min}_\pi(\varphi) = {\min}_\pi(\psi)$.
		It follows that $\varphi\lor{\min}_\pi(\varphi) \equiv \psi\lor{\min}_\pi(\psi)$.
		Since $\delta$ is a choice function, we conclude that
		$\delta_\pi(\varphi) = \delta(\varphi\lor{\min}_\pi(\varphi)) = \delta(\psi\lor{\min}_\pi(\psi)) =\delta_\pi(\psi)$.
	\end{description}

	Thus we have shown that $\delta_\pi$ is a choice function.
\end{proof}

Each composition generates a new choice function,
which in turn induces a rational contraction function that remains in the excerpt.

\begin{example}[continued from \cref{ex:cleaving-ltl}]
  Suppose we contract $\varphi \equiv p$ (\emph{``Mauricio swims today''}),
  and we have $\pi(n) = \psi_n$ for all $n$.
  We have $\minCleav{\pi}{p} = \psi_1$,
  as the formula $p \lor \psi_0$ is a tautology,
  whereas $\varphi \lor \psi_1 \equiv \psi_1$ (\emph{``if Mauricio swims tomorrow, he swims on at least two days''}),
  which is non-tautological.
  It follows that $\kb \dotmin[\delta_\pi] \varphi = \kb \dotmin[\delta] (\varphi\lor\psi_1)$.
  We contract \emph{``Mauricio swims today''} with~${\dotmin[\delta_\pi]}$ in the same way as we contract \emph{``if Mauricio swims tomorrow, he swims on at least two days''} with~$\dotmin[\delta]$.
\end{example}

Yet, the contraction functions induced by compositions are not necessarily computable.

\begin{toappendix}
	
\subsection{Uncomputability}
We prove the main result of this section: A non-finitary theory that admits any contraction must admit uncomputable contractions.

\begin{observation}
	\label{obs:comp-image}
	It is easy to see that $\mathrm{img}(\delta_\pi) \subseteq \mathrm{img}(\delta)$.
\end{observation}

\begin{lemma}
	\label{lem:diff-perm-diff-choice}
	Let $\pi,\pi' : \N\to\mathcal{C}$ be two \emph{distinct} permutations of $\mathcal{C}$.
	Then there exists a formula $\alpha \in \kb$ such that  $\delta_{\pi}(\alpha) \neq \delta_{\pi'}(\alpha)$.
\end{lemma}
\begin{proof}
	Since $\pi$ and $\pi'$ are different permutations,
	there must exist some indices $i,j,i',j'\in\N$ with $i<j$ and $i'<j'$
	such that $\pi'(i') = \pi(j)$ and $\pi'(j') = \pi(i)$.

	Consider now the formula $\alpha := \pi(i) \land \pi(j)$.
	Since $\pi(i),\pi(j)$ are in $\kb$, we clearly have $\alpha\in\kb$.

	As the next step, we show that ${\min}_\pi(\alpha) = \pi(i)$:
	\begin{itemize}
		\item Note that $\compl{\alpha} = \compl{\pi(i)} \cup \compl{\pi(j)}$.
		\item Since $\mathcal{C}$ does not contain a tautology, $\compl{\pi(i)}$ is non-empty,
		and hence $\compl{\alpha} \cap \compl{\pi(i)} \neq \emptyset$.
		\item Furthermore, the complements of $\pi(i)$ and any $\pi(k)$ with $k\neq i$ are disjoint (by property \textbf{(CL2)} of cleavings), and the same holds for $\pi(j)$.
		\item Hence, the only $k$ such that $\compl{\alpha}\cap\compl{k} \neq \emptyset$ are $k=i$ and $k=j$.
		\item Finally, recall that $i<j$.
	\end{itemize}
	It follows that ${\min}_\pi(\alpha) = \pi(i)$.
	Consequently, $\alpha \lor {\min}_\pi(\alpha) \equiv (\pi(i) \land \pi(j))\lor\pi(i) \equiv \pi(i)$,
	and thus it follows that $\delta_\pi(\alpha) = \delta(\pi(i)) \subseteq \compl{\pi(i)}$.
	Noting that $\alpha = \pi'(j') \land \pi'(i')$,
	and applying analogous reasoning,
	we have that $\delta_{\pi'}(\alpha) \subseteq \compl{\pi'(i')}= \compl{\pi(j)}$.

	Thus we have shown that $\delta_\pi(\alpha) \subseteq \compl{\pi(i)}$ and $\delta_{\pi'}(\alpha) \subseteq \compl{\pi(j)}$.
	With the disjointness of complements in a cleaving (\textbf{CL2}),
	it follows that $\delta_\pi(\alpha) \cap \delta_{\pi'}(\alpha) = \emptyset$.
	But since $\delta_\pi,\delta_{\pi'}$ are choice functions,
	and thus $\delta_\pi(\alpha)$ and $\delta_{\pi'}(\alpha)$ cannot be empty,
	we conclude that $\delta_\pi(\alpha)\neq \delta_{\pi'}(\alpha)$.
\end{proof}

\begin{lemma}
	\label{lem:diff-perm-diff-contr}
	Let $\pi,\pi' : \N\to\mathcal{C}$ be two \emph{distinct} permutations of $\mathcal{C}$.
	Then the induced contractions differ, i.e.,
	it holds that ${\dotdiv_{\delta_\pi}} \neq {\dotdiv_{\delta_{\pi'}}}$.
\end{lemma}
\begin{proof}
	By \cref{lem:diff-perm-diff-choice},
	there exists a formula $\alpha$ such that $\delta_\pi(\alpha) \neq \delta_{\pi'}(\alpha)$.
	Consider now the following sets of CCTs:
	$\decomp{\kb} \cup \delta_\pi(\alpha)$ and $\decomp{\kb} \cup \delta_{\pi'}(\alpha)$.
	Since $\decomp{\kb}$ contains only CCTs that contain $\alpha$,
	whereas $\delta_\pi(\alpha)$ and $\delta_{\pi'}(\alpha)$ contain only CCTs that do not contain $\alpha$,
	each of the two unions has no overlap.
	Therefore, we conclude that $\decomp{\kb} \cup \delta_\pi(\alpha) \neq \decomp{\kb} \cup \delta_{\pi'}(\alpha)$.
	With \textbf{(Compendious)}, it follows that
	\begin{align*}
\kb \dotdiv_{\delta_\pi} \alpha = \kb \cap \delta_\pi(\alpha)
	&= \bigcap (\decomp{\kb} \cup \delta_\pi(\alpha)) \\
	\bigcap (\decomp{\kb} \cup \delta_\pi(\alpha)) 	& \neq \bigcap (\decomp{\kb} \cup \delta_{\pi'}(\alpha)) \\
\bigcap (\decomp{\kb} \cup \delta_{\pi'}(\alpha))
	& = \kb \cap \delta_{\pi'}(\alpha) = \kb \dotdiv_{\delta_{\pi'}} \alpha
	\end{align*}
	Since $\dotdiv_{\delta_{\pi}}$ and $\dotdiv_{\delta_{\pi'}}$ differ on $\alpha$,
	they must be different contractions.
\end{proof}

\begin{lemma}
	\label{thm:uncountable-rat-remain}
	Let $\dotdiv$ be a rational contraction on a non-finitary theory $\kb$,
	such that $\dotdiv$ remains in the excerpt $\mathbb{E}$.
	Then there exist uncountably many rational contractions on $\kb$ that remain in $\mathbb{E}$.
\end{lemma}
\begin{proof}
	We have shown that $\kb$ contains an infinite cleaving (\cref{lem:nonfin-cleaving}).
	By \cref{lem:comp-choice-rational,lem:diff-perm-diff-contr},
	each permutation of this infinite cleaving induces a distinct rational contraction;
	and by \cref{obs:comp-image}, each of these contractions remains in $\mathbb{E}$.
	Since there are uncountably many permutations of an infinite set, the result follows.
\end{proof}

\begin{lemma}
\citep{jandson:towards-contraction} An ECF is fully rational iff its choice function satisfies both conditions:

\begin{description}
	\item[(C1)] $\delta(\varphi \land \psi) \subseteq \delta(\varphi) \cup \delta(\psi)$, for all formulae $\varphi$ and $\psi$;
	\item[(C2)] For all formulae $\varphi$ and $\psi$, if $\compl{\varphi} \cap \delta(\varphi\land\psi) \neq \emptyset$ then
	$\delta(\varphi) \subseteq \delta(\varphi \land \psi)$
\end{description}
\end{lemma}

\begin{lemma}
	\label{lem:comp-preserve-c1}
	For every permutation $\pi$, if $\delta$ satisfies \textbf{(C1)}, then so does $\delta_\pi$.
\end{lemma}
\begin{proof}
	Let $\varphi$ and $\psi$ be two formulae.
	We distinguish three cases:
	\begin{description}
		\item[Case 1: ${\min}_\pi(\varphi) = {\min}_\pi(\psi) = \bot$.]
		In this case, $\compl{\varphi}$ and $\compl{\psi}$ must both be disjoint from $\compl{\pi(k)}$ for all $k$.
		It follows that the same holds for $\compl{\varphi\land\psi}=\compl{\varphi}\cup\compl{\psi}$,
		and hence also ${\min}_\pi(\varphi\land\psi) = \bot$.
		We conclude:
		\[
		\delta_\pi(\varphi\land\psi) = \delta(\varphi\land\psi) \subseteq \delta(\varphi) \cup \delta(\psi) = \delta_\pi(\varphi) \cup \delta_\pi(\psi)
		\]
		\item[Case 2: ${\min}_\pi(\varphi) = \pi(i),{\min}_\pi(\psi)=\pi(j)$ for some $i,j$.]
		In this case, let us assume wlog.\ that $i\leq j$.
		It follows that both $\compl{\varphi}$ and $\compl{\psi}$ are disjoint from $\compl{\pi(k)}$ for all $k<i$.
		Consequently, $\compl{\varphi\land\psi} = \compl{\varphi} \cup \compl{\psi}$ is also disjoint from $\compl{\pi(k)}$ for all $k<i$;
		but is not disjoint from $\compl{\pi(i)}$.
		Hence, we have ${\min}_\pi(\varphi\land\psi) = \pi(i)$.

		It holds that $(\varphi\land\psi)\lor\pi(i) \equiv (\varphi\lor\pi(i))\land(\psi\lor\pi(i))$,
		and hence we conclude that
		\[ \delta_\pi(\varphi\land\psi) = \delta((\varphi\land\psi)\lor\pi(i)) = \delta((\varphi\lor\pi(i))\land(\psi\lor\pi(i))) \]
		If we now have $i=j$,
		then it follows (as $\delta$ satisfies \textbf{(C1)}) that
		\[
		\delta_\pi(\varphi\land\psi) = \delta(\varphi\lor\pi(i)) \cup \delta(\psi\lor\pi(i)) = \delta_\pi(\varphi) \cup \delta_\pi(\psi)
		\]
		and we are done.
		If $i<j$,
		then we must have $\compl{\psi}\cap\compl{\pi(i)} = \emptyset$,
		and $\psi\lor\pi(i)$ is a tautology.
		Hence,
		\[
		\delta_\pi(\varphi\land\psi) = \delta(\varphi\lor\pi(i)) = \delta_\pi(\varphi) \subseteq \delta_\pi(\varphi) \cup \delta_\pi(\psi)
		\]

		\item[Case 3: $\{{\min}_\pi(\varphi),{\min}_\pi(\psi)\} = \{\pi(i),\bot\}$ for some $i$.]
		In this case, let us assume wlog.\ that ${\min}_\pi(\varphi) = \bot$ and ${\min}_\pi(\psi) = \pi(i)$.
		Then $\compl{\varphi}$ is disjoint from $\compl{\pi(k)}$ for all $k$.
		Since $\compl{\varphi\land\psi} = \compl{\varphi} \cup \compl{\psi}$,
		it follows that ${\min}_\pi(\varphi\land\psi) = \pi(i)$.

		We know $\compl{\varphi}$ is disjoint from $\compl{\pi(i)}$,
		and thus $\varphi\lor\pi(i)$ is a tautology.
		It follows that $(\varphi\land\psi) \lor \pi(i) \equiv \psi\lor\pi(i)$.
		Hence,
		\[
		\delta_\pi(\varphi\land\psi) = \delta(\psi\lor\pi(i)) = \delta_\pi(\psi) \subseteq \delta_\pi(\varphi) \cup \delta_\pi(\psi)
		\]
	\end{description}
	Thus we have shown that $\delta_\pi$ indeed satisfies \textbf{(C1)}.
\end{proof}
\goodbreak

\begin{lemma}
	\label{lem:comp-preserve-c2}
	For every permutation $\pi$,
	if $\delta$ satisfies \textbf{(C2)}, then so does $\delta_\pi$.
\end{lemma}
\begin{proof}
	Let $\varphi$ and $\psi$ be formulae, and assume that $\compl{\varphi}\cap\delta_\pi(\varphi\land\psi)\neq\emptyset$.
	We distinguish four cases:
	\begin{description}
		\item[Case 1]: ${\min}_\pi(\varphi) = {\min}_\pi(\psi) = \bot$.
		In this case, $\compl{\varphi}$ and $\compl{\psi}$ must both be disjoint from $\compl{\pi(k)}$ for all $k$.
		It follows that the same holds for $\compl{\varphi\land\psi}=\compl{\varphi}\cup\compl{\psi}$,
		and hence also ${\min}_\pi(\varphi\land\psi) = \bot$.
		Then $\delta_\pi(\varphi) = \delta(\varphi)$, $\delta_\pi(\psi) = \delta(\pi)$ and $\delta_\pi(\varphi\land\psi) = \delta(\varphi\land\psi)$.
		Since $\delta$ satisfies \textbf{(C2)}, the result follows.

		\item[Case 2]: ${\min}_\pi(\varphi) = \pi(i)$, and either ${\min}_\pi(\psi) = \bot$ or ${\min}_\pi(\psi) = \pi(j)$ for some $j \geq i$.
		In this case, we observe that ${\min}_\pi(\varphi\land\psi) = \pi(i)$.

		Since $\varphi$ and $\pi(i)$ have shared complements, $(\varphi\land\psi)\lor\pi(i)$ cannot be a tautology.
		Then, by hypothesis and the fact that $\delta$ is a choice function,
		we have
		\begin{align*}
		\emptyset
		\neq \compl{\varphi} \cap \delta_\pi(\varphi\land\psi)
		& = \compl{\varphi} \cap \delta((\varphi\land\psi)\lor\pi(i)) \\
		& \subseteq \compl{\varphi} \cap \compl{(\varphi\land\psi)\lor\pi(i)} \\
		& \subseteq \compl{\varphi\lor\pi(i)}
		\end{align*}
		Then we have
		\begin{align*}
		\delta((\varphi\lor\pi(i))\land(\psi\lor\pi(i))) \cap \compl{\varphi\lor\pi(i)}
		\\
		\qquad = \delta((\varphi\land\psi)\lor\pi(i)) \cap \compl{\varphi\lor\pi(i)}
		\neq \emptyset.
				\end{align*}
		and by \textbf{(C2)} for $\delta$, we conclude that
		\[
		\delta_\pi(\varphi)
		= \delta(\varphi\lor\pi(i))
		\subseteq \delta((\varphi\lor\pi(i))\land(\psi\lor\pi(i)))
		= \delta_\pi(\varphi\land\psi)
		\]

		\item[Case 3]: ${\min}_\pi(\psi) = \pi(j)$ for some $j$, and either  ${\min}_\pi(\varphi) = \pi(i)$ for $i> j$ or ${\min}_\pi(\varphi) = \bot$.
		In this case, we observe that ${\min}_\pi(\varphi\land\psi) = \pi(j)$.
		Furthermore, it must hold that $\compl{\varphi}\cap\compl{\pi(j)} = \emptyset$.

		Since $\psi$ and $\pi(j)$ have shared complements, $(\varphi\land\psi)\lor\pi(j)$ cannot be a tautology.
		Then, by hypothesis and the fact that $\delta$ is a choice function, we have
		\begin{align*}
		\emptyset
		\neq \compl{\varphi}\cap\delta_\pi(\varphi\land\psi)
		& =  \compl{\varphi} \cap \delta((\varphi\land\psi)\lor\pi(j))
		\\
		& \subseteq \compl{\varphi} \cap \compl{(\varphi\land\psi)\lor\pi(j)}
		\\ & \subseteq \compl{\varphi} \cap \compl{\pi(j)}
		\end{align*}
		This means that $\compl{\varphi}\cap\compl{\pi(j)} \neq \emptyset$, and we have a contradiction.
	\end{description}
	Thus we have shown that $\delta_\pi$ indeed satisfies \textbf{(C2)}.¸
\end{proof}

\begin{lemma}\label{th:fully_remains_app}
	Let $\dotmin$ be a fully rational contraction on a non-finitary theory $\kb$,
	such that $\dotmin$ remains in the excerpt $\mathbb{E}$.
	Then there exist uncountably many fully rational contractions on $\kb$ that remain in $\mathbb{E}$.
\end{lemma}
\begin{proof}
	The proof proceeds analogously to the proof of \cref{thm:uncountable-rat-remain}.

	We have shown that $\kb$ contains an infinite cleaving (\cref{lem:nonfin-cleaving}).
	Each permutation of this infinite cleaving induces a distinct (\cref{lem:diff-perm-diff-contr}) fully (\cref{lem:comp-preserve-c1,lem:comp-preserve-c2}) rational (\cref{lem:comp-choice-rational}) contraction.
	By \cref{obs:comp-image}, each of these contractions remains in $\mathbb{E}$.
	Since there are uncountably many permutations of an infinite set, the result follows.
\end{proof}
\end{toappendix}

\begin{theoremrep}
  \label{thm:uncomputable}
	Let $\exc$ accommodate rational contraction, and let $\kb\in\exc$.
	The following statements are equivalent:
	\begin{enumerate}
		\item The theory $\kb$ is non-finitary.
		\item There exists an \emph{uncomputable} rational contraction function on $\kb$ that remains in $\exc$.
		\item There exists an \emph{uncomputable} fully rational contraction function on $\kb$ that remains in $\exc$.
	\end{enumerate}
\end{theoremrep}
\begin{inlineproof}[Proof Sketch]
	Let $\kb$ be non-finitary,
	and $\delta$ the choice function of a (fully) rational contraction for $\kb$ that remains in~$\exc$.
	Each permutation $\pi$ of a cleaving $\cleav\subseteq \kb$
	induces a \emph{distinct} (fully) rational contraction (with choice function $\delta_\pi$)
	that remains in $\exc$.
	At most countably many of these uncountably many (fully) rational contractions can be computable.

	If $\kb$ is finitary, every contraction function is computable,
	as it only has to consider finitely many formulae.
\end{inlineproof}
\begin{proof}
	We assume we can decide equivalence of formulae in the logic.
	\begin{itemize}
		\item (1) to (2): follows from \cref{thm:uncountable-rat-remain}.
		\item (2) to (3): follows from \cref{prop:accom-equiv,th:fully_remains_app}.
		\item (3) to (2) and (3) to (1): We show these by contraposition.
		  Fix a theory $\kb$ with finitely many equivalence classes, with representatives $\alpha_1,\ldots,\alpha_n$.
		  Then one can define a large class of (possibly non-rational) contractions as follows:
		  given $\varphi$, decide if $\varphi$ is equivalent to any non-tautological $\alpha_i$; if not, return some code word $w$ with $f(w) = \kb$; otherwise select some subset $X$ of $\{\alpha_1,\ldots,\alpha_n\}$ such that $\Cn{X}$ is in $\exc$, and return a code word $w$ with $f(w) =\Cn{X}$.

          The returned code word $w$ depends only on the representative $\alpha_i$ equivalent to $\varphi$, not on the syntax of $\varphi$.
          All of these functions are computable, and they include all (fully) rational contractions, so all (fully) rational contractions are computable.
	\end{itemize}

\end{proof}

\Cref{thm:uncomputable} makes evident that uncomputability of AGM contraction is inevitable. %
In fact, there are uncountably many uncomputable contraction functions.
Attempting to avoid this uncomputability by restraining the expressiveness of the excerpt leaves only the most trivial case: finitary theories.

\section{Effective Contraction in the B\"uchi Excerpt}
\label{sec:effective-approach}

Despite the strong negative result of \cref{sec:uncomputability},
computability can still be harnessed in particular excerpts:
excerpts~$\exc$ in which for every theory,
there exists at least one computable (fully) rational contraction function that remains in $\exc$.
We say that such an excerpt $\exc$ \emph{effectively accommodates} (fully) rational contraction.
If belief contraction is to be computed for compendious logics,
it is paramount to identify such excerpts as well as classes of computable contraction functions.
In this section, we show that the B\"uchi excerpt of LTL effectively accommodates (fully) rational contraction,
and we present classes of computable contraction functions.

For a contraction on a theory $\kb \in \excBuchi$ to remain in the B\"uchi excerpt,
the underlying choice function must be designed such that the intersection of $\kb$ with the selected CCTs corresponds to the support of a B\"uchi automaton.
As CCTs and ultimately periodic traces are interchangeable (\cref{lem:cct-iso}),
and
the support of a B\"uchi automaton is determined by the CCTs corresponding to its accepted ultimately periodic traces (\cref{lem:buchi-support-cct}),
a solution
is to design a selection mechanism, analogous to choice functions, that picks a single B\"uchi automaton instead of an (infinite) set of CCTs.

\begin{toappendix}
  In \cref{sec:effective-approach}, we define a new selection mechanism for contractions that remain in the B\"uchi excerpt:
\end{toappendix}
\begin{definitionrep}[B\"uchi Choice Functions]
	A \emph{B\"uchi choice function} $\gamma$ maps each LTL formula to a single B\"uchi automaton, such that
	for all LTL formulae $\varphi$ and $\psi$,
	\begin{description}
		\item[(BF1)] the language accepted by $\gamma(\varphi)$ is non-empty;
		\item[(BF2)] $\gamma(\varphi)$ supports $\lnot\varphi$, if $\varphi$ is not a tautology; and
		\item[(BF3)] $\gamma(\varphi)$ and $\gamma(\psi)$ accept the same language, if $\varphi \equiv \psi$.
	\end{description}
\end{definitionrep}
\goodbreak

Conditions \prop{BF1} - \prop{BF3} correspond to the respective conditions \prop{CF1} - \prop{CF3}.
Each B\"uchi choice function induces a rational contraction function. %

\begin{toappendix}
  In order to show that this selection mechanism gives rise to rational contractions,
  we connect it to the choice functions (\cref{def:choice-fun}) underlying exhaustive contraction functions (\cref{def:ecf}).
  Recall from \cref{lem:cct-iso} that a certain kind of traces,
  the ultimately periodic traces ($\pi\in\UP$), correspond to complete consistent theories $\upCCT{\pi}$ of LTL.
  A B\"uchi choice function $\gamma$ thus induces a choice function
  that selects all CCTs corresponding to ultimately periodic traces in the chosen B\"uchi automaton's accepted language:
  \begin{definition}
    The choice function $\delta_\gamma$ induced by a B\"uchi choice function $\gamma$ is the function
  \begin{equation*}
     \delta_\gamma(\varphi) = \{\, \upCCT{\pi} \mid \pi \in \UP \cap \lang{\gamma(\varphi)} \,\}
  \end{equation*}
  \end{definition}
  \begin{lemma}
  \label{lem:buchi-choice-cf}
     If $\gamma$ is a B\"uchi choice function, then $\delta_\gamma$ is a choice function, i.e., satisfies \prop{CF1} - \prop{CF3}.
  \end{lemma}
  \begin{proof}
    Let $\gamma$ be a B\"uchi choice function, which by definition satisfies \prop{BF1} - \prop{BF3}.
    We show each of the required properties for $\delta_\gamma$.
    \begin{description}
    \item[\prop{CF1}:] By \prop{BF1}, the language of $\gamma(\varphi)$ is non-empty.
      Per a classical result, any B\"uchi automaton $\gamma(\varphi)$ that recognizes a non-empty language must accept at least one ultimately periodic trace $\pi$.
      Hence, we have $\upCCT{\pi}\in\delta_\gamma(\varphi)$, and $\delta_\gamma(\varphi)$ is non-empty.
    \item[\prop{CF2}:] Let $\varphi$ be non-tautological, i.e., $\varphi \notin\Cn{\emptyset}$.
      By \prop{BF2}, it follows that $\gamma(\varphi)$ supports $\lnot\varphi$.
      This means that every trace $\pi\in\lang{\gamma(\varphi)}$ satisfies $\lnot\varphi$.
      In particular, this holds for every ultimately periodic $\pi\in\lang{\gamma(\varphi)}$.
      For such a $\pi$, it then follows that $(\lnot\varphi)\in\upCCT{\pi}$, or equivalently, $\varphi\notin\upCCT{\pi}$, and $\upCCT{\pi}\in\compl{\varphi}$.
      We have thus shown that $\delta_\gamma(\varphi) \subseteq \compl{\varphi}$ holds.
    \item[\prop{CF3}:] Let $\varphi\equiv\psi$. By \prop{BF3}, it follows that $\lang{\gamma(\varphi)} = \lang{\gamma(\psi)}$ holds.
      It is then easy to see, from the definition of $\delta_\gamma$, that also $\delta_\gamma(\varphi) = \delta_\gamma(\psi)$ holds.
    \end{description}
    We conclude that $\delta_\gamma$ is indeed a choice function.
  \end{proof}
  \begin{corollary}
    \label{corr:bcf-ecf}
    The B\"uchi contraction function induced by $\gamma$ is an exhaustive contraction function,
    with the underlying choice function $\delta_\gamma$.
  \end{corollary}
  \begin{proof}
    This follows directly from \cref{lem:buchi-choice-cf}, and rewriting $\support{\gamma(\varphi)}$ as $\bigcap\delta_\gamma(\varphi)$, using \cref{lem:buchi-support-cct}.
  \end{proof}
  To see that B\"uchi contraction functions remain in the B\"uchi excerpt,
  we show the following general property for the support of B\"uchi automata.
  \begin{lemma}
    \label{lem:support-inter}
    Let $A_1,A_2$ be B\"uchi automata.
    Then it holds that $\support{A_1}\cap\support{A_2} = \support{A_1 \unionbuchi A_2}$.
  \end{lemma}
  \begin{proof}
    Let $\varphi \in \support{A_1}\cap\support{A_2}$.
    By definition of support, this means that $\pi_1 \models \varphi$ for each $\pi_1\in\lang{A_1}$, and $\pi_2\models\varphi$ for each $\pi_2\in\lang{A_2}$.
    It is easy to see that this is equivalent to the statement that $\pi \models \varphi$ for each $\pi\in\lang{A_1}\cup\lang{A_2}$.
    As $\lang{A_1\unionbuchi A_2} = \lang{A_1}\cup\lang{A_2}$, and by definition of support,
    the latter statement is in turn equivalent to $\varphi\in\support{A_1\unionbuchi A_2}$.
    We have thus shown the equality.
  \end{proof}
  \begin{corollary}
    \label{corr:bcf-remains}
    Let $\dotmin[\gamma]$ be a B\"uchi contraction function on a theory $\kb = \support{A}$,
    where $A$ is a B\"uchi automaton.
    The contraction $\dotmin[\gamma]$
    satisfies $\kb\dotmin[\gamma] \varphi = \support{A\unionbuchi \gamma(\varphi)}$
    if $\varphi\in \kb$ and $\varphi$ is not a tautology, or $\kb\dotmin[\gamma]\varphi=\support{A}$ otherwise.
    Hence, $\dotmin[\gamma]$ remains in the B\"uchi excerpt.
  \end{corollary}
  \begin{proof}
    This follows directly from \cref{def:contractp} and \cref{lem:support-inter}.
  \end{proof}

  We have thus shown that BCFs are rational and remain in the B\"uchi excerpt.
  Let us now consider the opposite direction.
  We make use of the fact that every rational contraction is an ECF induced by some choice function $\delta$.

  \begin{definition}
    Let $\kb = \support{A}$, for a B\"uchi automaton $A$,
    and let $\dotmin[\delta]$ be a rational contraction on $\kb$ that remains in the B\"uchi excerpt, induced by a choice function $\delta$.
    We define the B\"uchi choice function $\gamma_{\delta}$,
    such that for each formula $\varphi$:
    \[
      \gamma_{\delta}(\varphi) = \begin{cases}
        \ltlbuchi{\varphi} & \textbf{if } \varphi\in\Cn{\emptyset}\\
        \ltlbuchi{\lnot\varphi} & \textbf{if } \varphi \notin \kb\\
        A' & \text{\textbf{else}, where } \support{A'} = \bigcap\delta(\varphi)
      \end{cases}
    \]
  \end{definition}
  \begin{lemma}
  \label{lem:choice-bcf}
      Let $\kb = \support{A}$, for a B\"uchi automaton $A$,
      and let $\dotmin[\delta]$ be a rational contraction on $\kb$ that remains in the B\"uchi excerpt, induced by a choice function $\delta$.
    The function $\gamma_{\delta}$ is a well-defined B\"uchi choice function.
  \end{lemma}
  \begin{proof}
    First, we show that an automaton $A'$ as in the definition always exists.
    Let us thus assume that $\varphi\in\kb$ is non-tautological.
    Since $\dotmin$ remains in the B\"uchi excerpt, we know that $\kb\dotmin\varphi = \support{A''}$ for some B\"uchi automaton $A'$.
    We define $A'$ as a B\"uchi automaton that recognizes precisely the language $\lang{A''} \setminus \lang{A}$.
    Such an automaton can always be constructed from $A$ and $A''$.
    It follows that $A'$ is unique, up to language-equivalence of automata.

    Clearly, we have $\lang{A''} = \lang{A}\cup\lang{A'}$.
    By \cref{lem:buchi-support-cct}, this implies $\support{A} \cap \support{A'} = \support{A''} = \kb\dotmin\varphi = \kb\cap\bigcap\delta(\varphi)$.
    Since the support depends only on the language of an automaton, this implies $\kb \dotmin \varphi = \support{A''} = \support{A\unionbuchi A'}$.
    As $\support{A} = \kb$, and as the decomposition of $\kb$ is necessarily disjoint from $\delta(\varphi)$, it follows that $\support{A'} = \bigcap\delta(\varphi)$.

    It remains to examine the conditions \prop{BF1} - \prop{BF3}.
    \begin{description}
    \item[\prop{BF1}:] In the first two cases, it is again easy to see that $\gamma_{\delta}(\varphi)$ accepts a non-empty language.
      In the third case, since $\dotmin$ is rational, we have $\delta(\varphi)\neq\emptyset$, and consequently the language of $A'$ must be non-empty.
    \item[\prop{BF2}:] We assume $\varphi$ is non-tautological, so the first case is ruled out.
      In the second case, clearly $\ltlbuchi{\lnot\varphi}$ supports $\lnot\varphi$.
      In the third case, we know that $\delta$ satisfies \prop{CF2}, i.e., $\delta(\varphi)\subseteq\compl{\varphi}$.
      Thus we have $(\lnot\varphi) \in \bigcap\compl{\varphi} \supseteq \bigcap\delta(\varphi) = \support{A'}$.
    \item[\prop{BF3}:] This follows directly from the definition of $\gamma_\delta$ and the fact that $\delta$ satisfies \prop{CF3}.
    \end{description}
    Thus we have shown that $\gamma_\delta$ is a B\"uchi choice function.
  \end{proof}
\end{toappendix}

\begin{definition}[B\"uchi Contraction Functions]
\label{def:contractp}
	Let $\kb$ be a theory in the B\"uchi excerpt
	and let
	$\gamma$ be a B\"uchi choice function.   %
	The \emph{B\"uchi Contraction Function (BCF)} on $\kb$ induced by $\gamma$ is the function %
	\[
	\kb \dotmin[\gamma] \varphi = \begin{cases}
		\kb \cap \support{\gamma(\varphi)}
		&\textbf{if } \varphi \notin \Cn{\emptyset} \text{ and } \varphi \in \kb\\
		\kb & \textbf{otherwise}
	\end{cases}
	\]
\end{definition}

\begin{toappendix}
\end{toappendix}

All such contractions remain in the B\"uchi excerpt.
Indeed, one can observe that if $\kb=\support{A}$ for a B\"uchi automaton $A$,
it holds that $\kb \cap \support{\gamma(\varphi)} = \support{A \unionbuchi \gamma(\varphi)}$,
where $\unionbuchi$ denotes the union of B\"uchi automata (cf.~\cref{sec:ltl}).
The class of all rational contraction functions that remain in the B\"uchi excerpt corresponds exactly to the class of all BCFs. %

\begin{theoremrep}
	\label{cor:comp-basic-contraction}
	A contraction function $\dotmin$ on a theory ${\kb \in \excBuchi}$
	is rational and remains within the B\"uchi excerpt
	if and only if ${\dotmin}$ is a BCF.
\end{theoremrep}
\begin{proof}
  We have already shown that BCFs are rational (by \cref{corr:bcf-ecf} and \cref{thm:ecf-representation}) and remain in the B\"uchi excerpt (\cref{corr:bcf-remains}).

  For the opposite direction, let $\dotmin$ be a rational contraction on $\kb$ that remains in the B\"uchi excerpt.
  By \cref{thm:ecf-representation}, $\dotmin$ must be induced by a choice function $\delta$.
  We have shown in \cref{lem:choice-bcf} that the corresponding function $\gamma_\delta$ is a B\"uchi choice function.
  From the definition of BCFs and ECFs, it is easy to see that $\dotmin = \dotmin[\gamma_\delta]$.
\end{proof}
\begin{figure}[t]%
\begin{minipage}{0.5\linewidth}
$\ltlbuchi{\lnot\Globally\Finally p}$:\\
\begin{tikzpicture}[thick,node distance=1.25cm,inner sep=1.5,background rectangle/.style={fill=gray!15,rounded corners}, show background rectangle]
  \node[st] (p1) {$p_1$};
  \node[st,right of=p1,accepting] (p2) {$p_2$};
  \draw[<-] (p1) -- ++(left:0.75);
  \draw[->] (p1) -- node[lbl]{$\emptyset$} (p2);
  \draw[->] (p1) edge[loop above] node[lbl]{$\emptyset,\{p\}$} ();
  \draw[->] (p2) edge[loop above] node[lbl]{$\emptyset$} ();
\end{tikzpicture}

$\gamma(\Globally\Finally p)$:\\
\begin{tikzpicture}[thick,node distance=1.25cm,inner sep=1.5,background rectangle/.style={fill=gray!15,rounded corners}, show background rectangle]
  \node[st] (z1) {$z_1$};
  \node[st,right of=z1] (z2) {$z_2$};
  \node[st,right of=z2,accepting] (z3) {$z_3$};
  \draw[<-] (z1) -- ++(left:0.75);
  \draw[->] (z1) -- node[lbl] {$\{p\}$} (z2);
  \draw[->] (z2) edge[loop above] node[lbl] {$\emptyset,\{p\}$} ();
  \draw[->] (z2) -- node[lbl] {$\emptyset$} (z3);
  \draw[->] (z3) edge[loop above] node[lbl] {$\emptyset$} ();
\end{tikzpicture}
\end{minipage}
\hfill
\begin{minipage}{0.45\linewidth}
$A_\kb \sqcup \gamma(\Globally\Finally p)$:\\
\begin{tikzpicture}[thick,node distance=1.25cm,inner sep=1.5,background rectangle/.style={fill=gray!15,rounded corners}, show background rectangle]
  \node[st] (m1) {$z_1$};
  \node[st,right of=m1] (m2) {$z_2$};
  \node[st,right of=m2,accepting] (m3) {$z_3$};
  \node[st,above of=m1] (m4) {$q_0$};
  \node[st,right of=m4,accepting] (m5) {$q_1$};
  \node[st,right of=m5,accepting] (m6) {$q_2$};

  \draw[<-] (m1) -- ++(left:0.75);
  \draw[->] (m1) -- node[lbl]{$\{p\}$} (m2);
  \draw[->] (m2) edge[loop below] node[lbl]{$\emptyset,\{p\}$} ();
  \draw[->] (m2) -- node[lbl]{$\emptyset$} (m3);
  \draw[->] (m3) edge[loop below] node[lbl]{$\emptyset$} ();
  \draw[<-] (m4) -- ++(left:0.75);¸
  \draw[->] (m4) edge[loop above] node[lbl]{$\emptyset,\{p\}$} ();
  \draw[->] (m4) -- node[lbl]{$\emptyset,\{p\}$} (m5);
  \draw[->] (m5) edge[bend left] node[lbl]{$\{p\}$} (m6);
  \draw[->] (m6) edge[bend left] node[lbl]{$\emptyset,\{p\}$} (m5);
\end{tikzpicture}
\end{minipage}%
\caption{
  BCF contraction of $\Globally\Finally p$ from $\support{A_\kb}$. %
}%
\label{fig:buchi-contract-basic}%
\end{figure}%
\begin{example}
\label{ex:buchi-contract-basic}
  Let $\kb = \support{A_\kb}$, for the B\"uchi automaton $A_\kb$ shown in \cref{fig:buchi-support}.
  To contract the formula $\Globally \Finally p$,
  a B\"uchi choice function $\gamma$ may select the B\"uchi automaton $\gamma(\Globally \Finally p)$
  shown in \cref{fig:buchi-contract-basic}.
  This automaton supports $\lnot\Globally\Finally p$; the automaton $\ltlbuchi{\lnot\Globally\Finally p}$ is shown for reference.
  In fact, $\gamma(\Globally \Finally p)$ accepts precisely the traces satisfying ${p\land\lnot\Globally\Finally p}$.
  In our swimming example (cf.~\cref{ex:ltl-swim}),
  this corresponds to \emph{``Mauricio swims today, but does not swim infinitely often.''}

  The result of the contraction is
  the belief state ${\support{A_\kb \unionbuchi \gamma(\Globally\Finally p)}}$, whose supporting automaton is also shown in \cref{fig:buchi-contract-basic}.
  The union $\sqcup$ is obtained by simply taking the union of states and transitions.
  This automaton does not support $\Globally \Finally p$, and therefore the contraction is successful.
  The other supported formulae listed in \cref{fig:buchi-support} are still supported (see \cref{ex:support-swim} for a discussion of their meaning).
\end{example}

As BCFs capture all rational contractions within the excerpt, it follows from \Cref{thm:uncomputable} that not all BCFs are computable.
Note from \cref{def:contractp} that to achieve computability, it suffices to be able to: %
(i)~decide if $\varphi$ is a tautology,
(ii)~decide if $\varphi \in \kb $,
(iii)~compute the underlying B\"uchi choice function $\gamma$, and
(iv)~compute the intersection of $\kb$ with the support of $\gamma(\varphi)$. %
Conditions~(i) and (ii) can be realised with standard reasoning methods for LTL and B\"uchi automata~\citep{clarke:model-checking}.
For condition~(iv), we observe above that the intersection of the support of two automata is equivalent to the support of their union.
As $\gamma$ produces a B\"uchi automaton, and union of B\"uchi automata is computable, condition~(iv) is also satisfied. Therefore, condition~(iii) is the only one remaining.
It turns out that (iii) is a necessary and sufficient condition to characterise all computable contraction functions within the B\"uchi excerpt.

\begin{theoremrep}
  \label{thm:contractp-compute}
	Let $\dotmin$ be a rational contraction function on a theory $\kb \in \excBuchi$, such that $\dotmin$ remains in the B\"uchi excerpt.
	The operation ${\dotmin}$ is computable iff
	${\dotmin} = {\dotmin[\gamma]}$ for some computable B\"uchi choice function $\gamma$ .
\end{theoremrep}
\begin{proof}
  Let $\kb=\support{A}$ for a B\"uchi automaton $A$, and let $\gamma$ be a computable B\"uchi choice function.
  Recall that it is decidable whether a given $\varphi$ is tautological (by deciding equivalence with $\top$),
  and whether $\varphi \in \kb$ (\cref{thm:buchi-decide}).
  Further, if $\varphi$ is neither tautological nor absent from $\kb$,
  we have $\kb\dotmin\varphi = \support{A \unionbuchi \gamma(\varphi)}$, and there exists an effective construction for the union operator $\unionbuchi$.
  It is then easy to see from \cref{def:contractp} that the BCF $\dotmin[\gamma]$ is computable.

  To see that computability of the B\"uchi choice function is necessary,
  suppose a given contraction $\dotmin$ is computable, rational, and remains in the B\"uchi excerpt.
  By rationality, $\dotmin$ is an ECF induced by some choice function $\delta$.
  Then it is easy to see that $\dotmin = \dotmin[\gamma_\delta]$, for the B\"uchi choice function $\gamma_\delta$.
  And in fact, this B\"uchi choice function $\gamma_\delta$ can be computed:
  We can decide which case is applicable (again, by decidability of tautologies and membership in $\kb$),
  and the respective automata can be constructed.
  Of particular interest, in the third case, we can construct the automaton $A'$ as the difference of the automaton $A''$ supporting $\kb\dotmin\varphi$ (which is computable) and the automaton $A$
  (cf.\ the proof of \cref{lem:choice-bcf}).
\end{proof}

In the following, we define a large class of computable B\"uchi choice functions.
As outlined in \cref{sec:contraction}, a choice function is an extra-logical mechanism that realises the epistemic preferences of an agent,
which can be formalised as a preference relation on CCTs.
Due to the tight connection between CCTs and ultimately periodic traces,
we can equivalently formalise the epistemic preferences as a relation on ultimately periodic traces.
To attain computability, we finitely represent such a (potentially infinite) relation on traces
using a special kind of B\"uchi automata:

\begin{toappendix}
  Towards computable, relational choice functions, we represent the epistemic preferences of an agent using a special kind of B\"uchi automata:
\end{toappendix}

\begin{definitionrep}[B\"uchi-Mealy Automata]
	A \emph{B\"uchi-Mealy automaton}
	is a B\"uchi automaton on %
	$\Sigma_\textnormal{BM} = \powerset{\AP} \times \powerset{\AP}$.
\end{definitionrep}
A B\"uchi-Mealy automaton $B$ accepts infinite sequences of pairs $(a_1, b_1)(a_2,b_2)\cdots (a_i,b_i)\cdots $ with  ${a_i, b_i \in \powerset{\AP}}$, for all $i \geq 1$.
Such an infinite sequence corresponds to a pair of  traces $(\pi_1, \pi_2)$ where $\pi_1 = a_1a_2\cdots a_i \cdots$ and $\pi_2 = b_1b_2\cdots b_i \cdots$. Therefore, a B\"uchi-Mealy automaton~$B$ recognises the  binary relation %
\begin{equation*}
	\!\!\!\!\!
	\rel{B}\!:=\!\big\{ (a_1\cdots , b_1\cdots)
	\!\mid\! (a_1,b_1)(a_2,b_2)\cdots \in \lang{B} \big\}
	\label{eq:buchi-mealy-rel}
\end{equation*}
If $(\pi_1, \pi_2) \in \rel{B}$ then $\pi_2$ is at least as plausible as $\pi_1$.

\begin{figure}[t]
	\begin{minipage}{0.27\linewidth}
    $B$:\\ \resizebox{\linewidth}{!}{%
   	\begin{tikzpicture}[thick,node distance=1.35cm,inner sep=1.5,baseline=(q0),background rectangle/.style={fill=gray!15,rounded corners},inner frame xsep=0.2ex, show background rectangle]
   		\node[st] (q0) {$q_0$};
   		\node[st,accepting,right of=q0] (q1) {$q_1$};
   		\draw[<-] (q0) -- ++(left:0.6);
   		\draw[->] (q0) edge[loop above] node[lbl]{$\emptyset / \emptyset$} ();
   		\draw[->] (q0) -- node[lbl]{$\emptyset / \{p\}$} (q1);
   		\draw[->] (q1) edge[loop above] node[lbl,align=center]{$\Sigma / \Sigma$} ();
   	\end{tikzpicture}}%
	\end{minipage}
	\hfil
	\begin{minipage}{0.6\linewidth}
	accepting run:
	\newcommand\step[2]{\xrightarrow{\mathmakebox[1.5em]{\begin{array}{c}#1\\/\\#2\end{array}}}}
	\[
	  q_0 \step{\emptyset}{\emptyset}
	  q_0 \step{\emptyset}{\{p\}}
	  q_1 \step{\{p\}}{\emptyset}
	  q_1 \step{\{p\}}{\emptyset}
	  \cdots
	\]
	\end{minipage}

	\caption{
		A B\"uchi-Mealy automaton $B$ on $\AP = \{p\}$.
		By convention, we write $a/b$ rather than $(a,b)$.
		A label containing $\Sigma$ denotes transitions for both $\emptyset$ and $\{p\}$.
		On the right, an accepting run of $B$.%
}
	\label{fig:buchi-mealy}
\end{figure}
\begin{example}
	\label{ex:mealy-first}
	Consider again the swimming example (cf.\ \cref{ex:ltl-swim}),
	and an epistemic preference that deems scenarios in which Mauricio swims later to be less plausible than those where he swims sooner.
	This preference is expressed by the B\"uchi-Mealy automaton~$B$ shown in \cref{fig:buchi-mealy} (on the left).
	The automaton~$B$ recognises the relation
	\[
	\rel{B} = \{\,(\pi,\pi') \in \Sigma^\omega \times \Sigma^\omega \mid \mathit{first}_p(\pi) > \mathit{first}_p(\pi') \,\}
	\]
	where $\mathit{first}_p(\cdot)$ is the index of the first occurrence of proposition $p$ in the given trace (and $\infty$ if $p$ never occurs).
	In other words, the earlier $p$ occurs in a trace, the more plausible is such a trace.
	For instance,
	the accepting run on the right-hand side of \cref{fig:buchi-mealy} is the reason that
	the trace $\emptyset\emptyset\{p\}^\omega$ (where Mauricio swims only in two days)
	is considered less plausible than $\emptyset\{p\}\emptyset^\omega$ (where Mauricio already swims tomorrow),
	and hence the pair $(\emptyset\emptyset\{p\}^\omega, \emptyset\{p\}\emptyset^\omega)$ is in $\rel{B}$.
\end{example}%
An epistemic preference relation induces a choice function which always selects the \emph{maximal}, i.e., the most plausible CCTs that do not contain the given formula.
In order to analogously define the B\"uchi choice function induced by a B\"uchi-Mealy automaton,
we show that the set of most plausible CCTs can be represented by a B\"uchi automaton.
\begin{toappendix}
	
In order to construct computable choice functions, we apply various automata constructions to the given formula to be contracted and the B\"uchi-Mealy automaton representing the epistemic preferences.
This allows us to satisfy the syntax-insensitivity condition \prop{BF3}:
Since equivalent LTL formula $\varphi,\psi$ are satisfied by the same traces,
the respective automata $\ltlbuchi{\varphi}$ and $\ltlbuchi{\psi}$ recognize the same languages.
If we compute the choice function $\gamma$ purely by applying automata operations that preserve language equivalence,
condition \prop{BF3} is naturally satisfied.

The classical automata constructions for union, intersection and complementation of B\"uchi automata can be directly applied to B\"uchi-Mealy automata, as they are a special case.
However, beyond these classical boolean operators, we require operations that reflect the relational nature of B\"uchi-Mealy automata.
In particular, we use the following two constructions to convert between normal B\"uchi automata (which represent sets) and B\"uchi-Mealy automata (which represent relations):
\begin{lemma}
  \label{lem:bm-prod}
  Given a B\"uchi automaton $A$ over an alphabet $\Sigma$,
  one can effectively construct a B\"uchi-Mealy automaton recognizing the relation $\Sigma^\omega\times\lang{A}$.
\end{lemma}
\begin{proof}
  Let $A = (Q, \Sigma, \Delta, Q_0, R)$.
  We construct the B\"uchi-Mealy automaton as $B = (Q, \Sigma \times \Sigma, \Delta_B, Q_0, R)$,
  where $\langle q, a/b, q'\rangle \in \Delta_B$ iff $\langle q,b,q'\rangle \in \Delta$ and $a\in \Sigma$.

  For any accepted pair $(a_1a_2\ldots, b_1b_2\ldots) \in \rel{B}$,
  we have that $b_1b_2\ldots \in \lang{A}$, with the same accepting run.
  Hence $\rel{B} \subseteq \Sigma^\omega \times \lang{A}$.

  Conversely, any accepting run of $A$ for a word $b_1b_2\ldots \in \lang{A}$
  is an accepting run for $(a_1,b_1)(a_2,b_2)\ldots$ in $B$,
  for any word $a_1a_2\ldots\in\Sigma^\omega$.
  Hence $\Sigma^\omega \times \lang{A} \subseteq \rel{B}$.
\end{proof}

For a relation $R$, let $\operatorname{proj}_1$ denote the projection of a relation $R$ to the set of first components in each pair in the relation,
i.e., $\operatorname{proj}_1(R) = \{\, x \mid \text{there exists } y \text{ s.t.\ } (x,y)\in R\,\}$.

\begin{lemma}
  \label{lem:bm-proj}
  Given a B\"uchi-Mealy automaton $B$ over an alphabet $\Sigma$,
  one can effectively construct a B\"uchi automaton recognizing the language $\operatorname{proj}_1\big(\rel{B}\big)$.
\end{lemma}
\begin{proof}
  Let $B = (Q, \Sigma \times \Sigma, \Delta, Q_0, R)$.
  We construct the B\"uchi automaton $A = (Q, \Sigma, \Delta_A, Q_0, R)$,
  where $\langle q, a, q' \rangle \in \Delta_A$ iff there exists some $b\in\Sigma$ such that $\langle q, a/b, q' \rangle \in \Delta$.

  Any accepting run of $A$ for a word $a_1a_2\ldots \in \lang{A}$
  is an accepting run of $B$ for a word of the form $(a_1a_2\ldots, b_1b_2\ldots) \in \rel{B}$ for some $b_1b_2\ldots\in\Sigma^\omega$.
  Hence we have $\lang{A} \subseteq \operatorname{proj}_1(\rel{B})$.

  Conversely, let $a_1a_2\ldots \in \operatorname{proj}_1\big(\rel{B}\big)$.
  Then there exists some $b_1b_2\ldots \in \Sigma^\omega$
  such that there is an accepting run for $(a_1, b_1)(a_2, b_2) \ldots$ in $B$.
  This run is an accepting run for $a_1a_2\ldots$ in $A$.
  We conclude that $\operatorname{proj}_1\big(\rel{B}\big) \subseteq \lang{A}$.
\end{proof}

\end{toappendix}
\begin{lemmarep}
\label{lem:bm-max-buchi}
  Let $B$ be a B\"uchi-Mealy automaton, and $\varphi$ an LTL formula.
  There exists a B\"uchi automaton $\maxaut{B}{\varphi}$ such that
  $
    \lang{\maxaut{B}{\varphi}} = {\max}_{\rel{B}} \{\, \pi\in\Sigma^\omega \mid \pi\models\varphi \,\}
  $.
\end{lemmarep}
\begin{proof}
  Recall that the maximal elements of a set $X$ wrt.\ a relation $R$
  are defined as
  \[
    {\max}_R(X) := \{\, x\in X\mid \text{there is no } y\in X \text{ s.t.\ } (x,y) \in R \,\}
  \]
  Let us assume that $R$ is a relation over $\Sigma^\omega$.
  We observe, through set-theoretic reasoning:
  \begin{align*}
    {\max}_R(X) &= \{\, x\in X\mid \lnot \exists y\in X\,.\, (x,y) \in R \,\}\\
      &= X \setminus \{\, x \in \Sigma^\omega \mid \exists y \in X \,.\, (x, y) \in R \,\}\\
      &= X \setminus \{\, x \in \Sigma^\omega \mid \exists y\in \Sigma^\omega\,.\, (x,y) \in R \cap (\Sigma^\omega \times X) \,\}\\
      &= X \setminus \operatorname{proj}_1\big(R \cap (\Sigma^\omega \times X)\big)
  \end{align*}
  We instantiate $X$ with the set of traces that satisfy $\varphi$,
  or equivalently, with $\lang{\ltlbuchi{\varphi}}$.
  We instantiate the relation $R$ with the epistemic preference relation $\rel{B}$.
  Every operation in the term $\lang{\ltlbuchi{\varphi}} \setminus \operatorname{proj}_1(\rel{B} \cap (\Sigma^\omega \times \lang{\ltlbuchi{\varphi}}))$
  can be implemented as an automata operation on B\"uchi automata resp.\ B\"uchi-Mealy automata.
  In particular, this is also the case for the cartesian product $\Sigma^\omega \times \lang{\ltlbuchi{\varphi}}$ (cf.\ \cref{lem:bm-prod})
  as well as for the projection $\operatorname{proj}_1$ (cf.\ \cref{lem:bm-proj}).
  The result of these operations is a B\"uchi automaton $\maxaut{B}{\varphi}$ recognizing precisely the language ${\max}_{\rel{B}}\{\,\pi\in\Sigma^\omega\mid\pi\models\varphi\,\}$.
\end{proof}
The B\"uchi choice function induced by a B\"uchi-Mealy automaton $B$ is the function $\gamma_B$ with $\gamma_B(\varphi) = \maxaut{B}{\lnot\varphi}$, if $\varphi$ is non-tautological, and $\gamma_B(\varphi)=\ltlbuchi{\varphi}$ otherwise.
The automaton $\maxaut{B}{\lnot\varphi}$ can be constructed from $B$ and $\varphi$ through a series of effective automata constructions,
as detailed in the proof of \cref{lem:bm-max-buchi}.
Consequently,
$\gamma_B$ is computable.
\begin{toappendix}
  As discussed in \cref{sec:contraction},
  in order to define a choice function from an epistemic preference relation $<$,
  it is required that the preference relation satisfies a condition called \prop{Maximal Cut}:
  for each formula $\varphi$, the set ${\max}_{<}(\compl{\varphi})$ must be non-empty.
  A similar condition also applies when defining a B\"uchi choice function from an epistemic preference relation represented as a B\"uchi-Mealy automaton $B$.
  We accordingly reformulate the condition:
  \begin{description}
  \item[(Maximal Cut)] For every non-tautological $\varphi\in\LTL$, it must hold that ${\max}_{\rel{B}}\{\,\pi\in\Sigma^\omega\mid\pi\models\lnot\varphi\,\} \neq \emptyset$.
  \end{description}

  In order to connect with the results by \citet{jandson:towards-contraction},
  we define the epistemic preference relation $<_B$ (on CCTs) induced by a B\"uchi-Mealy automaton $B$ as the relation
  \[
  {<_B} = \{\,(\upCCT{\pi_1}, \upCCT{\pi_2}) \mid (\pi_1,\pi_2)\in\rel{B}\,\}
  \]

  Though a B\"uchi-Mealy automaton $B$ recognizes a relation over general traces,
  we are thus only interested in the relation between ultimately periodic traces (which represent CCTs).
  However, it turns out that the two are interchangeable:
  the ultimately periodic traces among the maximal elements (of some set of traces $\lang{A}$) wrt.\ the relation $\rel{B}$
  are exactly the maximal elements of $\lang{A}\cap\UP$ wrt.\ $\rel{B}$. %

  \begin{lemma}
    Let $A$ be a B\"uchi automaton, and $B$ a B\"uchi-Mealy automaton.
    Then it holds that
    \[
      {\max}_{\rel{B}}(\lang{A}) \cap \UP = {\max}_{\rel{B}}(\lang{A} \cap \UP)
    \]
  \end{lemma}
  \begin{proof}
    It is easy to see that a trace in ${\max}_{\rel{B}}(\lang{A}) \cap \UP$ is also maximal in $\lang{A}\cap \UP$.
    For the opposite direction, suppose some trace $\pi\in {\max}_{\rel{B}}(\lang{A} \cap \UP)$ were non-maximal in $\lang{A}$,
    i.e., there existed some $\pi' \in \lang{A}$ with $(\pi,\pi')\in\rel{B}$.
    We can construct a B\"uchi automaton that recognizes all such traces $\pi'$,
    i.e., an automaton $A_\pi$ such that
    \begin{align*}
      \lang{A_\pi} &= \{\,\pi'\in\lang{A}\mid (\pi,\pi')\in\rel{B}\,\}\\
      &=\lang{A} \cap \operatorname{proj}_2(\rel{B} \cap (\{\pi\}\times \Sigma^\omega))
    \end{align*}
    Similar to the proof of \cref{lem:bm-max-buchi},
    we note that every operation in this characterization of $\lang{A_\pi}$ can be implemented as an automata operation:
    $\operatorname{proj}_2$ works analogously to $\operatorname{proj}_1$ (\cref{lem:bm-proj}), and the cartesian product $\cdot \times \Sigma^\omega$ is analogous to the construction of \cref{lem:bm-prod}.
    Since $\pi$ is ultimately periodical, we can also construct a B\"uchi automaton that recognizes exactly the singleton language $\{\pi\}$.

    By assumption, the language of $A_\pi$ is non-empty.
    As any B\"uchi automaton with a non-empty language accepts at least one ultimately periodical trace,
    there must in particular exist an ultimately periodic trace $\pi''\in\lang{A_\pi}$, i.e., a trace $\pi'' \in \lang{A}\cap\UP$ such that $(\pi,\pi'')\in\rel{B}$.
    But this contradicts the assumption that $\pi\in {\max}_{\rel{B}}(\lang{A} \cap \UP)$.
    Hence, our assumption was incorrect, and the equality holds.
  \end{proof}

  \begin{lemma}
    \label{lem:bm-max-compl}
    Let $B$ be a B\"uchi-Mealy automaton.
    For every formula $\varphi$,
    we have that $\support{\maxaut{B}{\lnot\varphi}} = \bigcap {\max}_{<_B}(\compl{\varphi})$.
  \end{lemma}
  \begin{proof}
    Recall that, by \cref{lem:bm-max-buchi}, we have
    \begin{align*}
      \lang{\maxaut{B}{\lnot\varphi}} &= {\max}_{\rel{B}} \{\,\pi\in\Sigma^\omega\mid\pi\models\lnot\varphi\,\}\\
      &= {\max}_{\rel{B}}\lang{\ltlbuchi{\lnot\varphi}}
    \end{align*}
    With \cref{lem:buchi-support-cct},
    it follows that
    \begin{align*}
      \support{\maxaut{B}{\lnot\varphi}} &= \bigcap \{\,\upCCT{\pi}\mid\pi\in{\max}_{\rel{B}}(\lang{\ltlbuchi{\lnot\varphi}}) \cap \UP \,\}\\
      &= \bigcap \{\,\upCCT{\pi}\mid\pi\in{\max}_{\rel{B}}(\lang{\ltlbuchi{\lnot\varphi}} \cap \UP)  \,\}
    \end{align*}
    As the ultimately periodic traces that satisfy $\lnot\varphi$ correspond precisely to the CCTs that contain $\lnot\varphi$,
    i.e., the complements $\compl{\varphi}$ of $\varphi$,
    we have
    $
      \support{\maxaut{B}{\lnot\varphi}} = \bigcap {\max}_{<_B}(\compl{\varphi})
    $.
  \end{proof}
  A consequence of this lemma is that the relation $<_B$ satisfies \prop{Maximal Cut} if and only if $\rel{B}$ satisfies our reformulation of the property above, justifying the naming.
\end{toappendix}
\begin{propositionrep}
	\label{prop:buchi-mealy-choice}
	If the relation $\rel{B}$ recognised by a
	B\"uchi-Mealy automaton $B$ satisfies \prop{Maximal Cut}, then %
	$\gamma_B$ is a computable B\"uchi choice function.
\end{propositionrep}
\begin{proof}
  The fact that $\gamma_B$ is computable follows from the effectiveness of the construction in the proof of \cref{lem:bm-max-buchi},
  and the fact that it is decidable whether $\varphi$ is a tautology.

  It remains to see that $\gamma_B$ is a B\"uchi choice function.
  \begin{description}
  \item[\prop{BF1}:] From the definition of $\gamma_B$, and \prop{Maximal Cut}, it follows that $\gamma_B(\varphi)$ always recognizes a non-empty language.
  \item[\prop{BF2}:] If $\varphi$ is not a tautology, we have $\gamma_B(\varphi) = \maxaut{B}{\lnot\varphi}$. By \cref{lem:bm-max-buchi}, this automaton accepts only traces that satisfy $\lnot\varphi$, and hence supports $\lnot\varphi$.
  \item[\prop{BF3}:] If $\varphi \equiv \psi$, either both are tautologies, and we have $\lang{\gamma_B(\varphi)} = \lang{\ltlbuchi{\varphi}} = \Sigma^\omega = \lang{\ltlbuchi{\psi}} = \lang{\gamma_B(\psi)}$;
  or neither of them are tautologies, and we have $\lang{\gamma_B(\varphi)} = \lang{\maxaut{B}{\varphi}} = \lang{\maxaut{B}{\psi}} = \lang{\gamma_B(\psi)}$.
  \end{description}
  Thus $\gamma_B$ is indeed a computable B\"uchi choice function.
\end{proof}

To obtain fully rational computable contraction functions,
it suffices that the relation recognised by the B\"uchi-Mealy automaton satisfies \prop{Mirroring} as well as \prop{Maximal Cut}.
\begin{theoremrep}
	Let $\kb$ be a theory in the B\"uchi excerpt, and let $B$ be a B\"uchi-Mealy automaton
	such that the relation $\rel{B}$ satisfies \prop{Mirroring} and \prop{Maximal Cut}.

	The BCF~$\dotmin[\rel{B}]$ is fully rational and computable.
\end{theoremrep}
\begin{proof}
  Computability follows from \cref{thm:contractp-compute,prop:buchi-mealy-choice}.
  For full rationality, we conclude from \cref{lem:bm-max-compl}, \prop{Maximal Cut} and \prop{Mirroring}, that $\dotmin[\rel{B}]$ is a Blade Contraction Function.
  With \cref{thm:bcf-representation}, the result follows.
\end{proof}

\begin{example}[continued from \cref{ex:mealy-first}]
\label{ex:buchi-contract-relational}
  Consider the B\"uchi automaton $A_\kb$ in \cref{fig:buchi-support},
  and the epistemic preference expressed by the B\"uchi-Mealy automaton $B$ in \cref{fig:buchi-mealy}.
  Note that the relation $\rel{B}$ satisfies both
  \prop{Maximal Cut}, as there always exists an earliest-possible occurrence of~$p$,
  and \prop{Mirroring}.
  To contract the formula $\varphi :\equiv \Globally\Finally p$ (\emph{``Mauricio swims infinitely often''}) from $\support{A_\kb}$,
  we construct the automaton~$\maxaut{B}{\lnot\varphi}$ representing only the most plausible CCTs. This automaton is equivalent to the automaton~$\gamma(\Globally\Finally p)$ shown in \cref{fig:buchi-contract-basic}.
  The most preferrable traces wrt.\ $\rel{B}$ are those where $p$ holds already in the first step (\emph{``Mauricio swims today''}).
  Therefore, the result of the contraction is the same as in \cref{ex:buchi-contract-basic}.
\end{example}

As there exist B\"uchi-Mealy automata that satisfy \prop{Mirroring} and \prop{Maximal Cut}, such as the automaton discussed in \cref{ex:mealy-first,ex:buchi-contract-relational},
we conclude that the B\"uchi excerpt effectively accommodates fully rational contraction.

\section{Conclusion}
\label{sec:conclusion}

We have investigated the computability of  AGM contraction  for the class of compendious logics, which embrace several logics used in computer science and AI.
Due to the high expressive power of these logics,  not all epistemic states admit a finite representation. %
Hence, the epistemic states that an agent can assume are confined to a space of theories, which depends on a method of finite representation.   %
We have shown a severe negative result:  no matter which form of finite representation we use, as long as it  does not collapse to the finitary case, AGM contraction suffers from uncomputability.
Precisely, there are uncountably many uncomputable (fully) rational contraction functions in all such expressive spaces.
This negative result also impacts other forms of belief change.
For instance, in the presence of classical negation, revision and  contraction are interdefinable via Levi  and Harper identities~\citep{jandson:thesis}. Thus, it is likely that revision also suffers from uncomputability.
Accordingly, uncomputability might span to iterated belief revision \citep{DarwicheP97}, update and erasure  \citep{KMUP},  and pseudo-contraction \citep{Hansson93}, to cite a few.
It is worth investigating uncomputability of these other operators.

In this work, we have focused on the AGM paradigm, and logics which are Boolean. We intend to expand our results for a wider class of logics by dispensing with the Boolean operators, and assuming only that the logic is AGM compliant.
We believe the results shall hold in the more general case, as our negative results follow from cardinality arguments.
On the other hand, several logics used in knowledge representation and reasoning are not AGM compliant, as for instance a variety of description logics \citep{wassermann:relevance-recovery}.
In these logics, the \emph{recovery} postulate \cp[5] can be replaced by the \emph{relevance} postulate \citep{HanssonRelevance}, and contraction functions can be properly defined.
Such logics are called relevance-compliant.
As relevance is an weakened version of recovery, the uncomputability results in this work translate to various relevance-compliant logics. However, it is unclear if all such logics are affected by uncomputability. We aim to investigate this issue in such logics.

Even if we have to coexist with uncomputability, we can still identify classes of operators which are guaranteed to be computable.
To this end, we have introduced a novel class of computable contraction functions for LTL using B\"uchi automata.
This is an initial step towards the application of belief change in other areas, such as methods for automatically repairing systems \citep{GuerraW18}.
The methods devised here for LTL form a foundation for the development of analogous strategies for other expressive logics,
such as CTL, $\mu$-calculus and many description logics.
For example, in these logics, similarly to LTL, decision problems such as satisfiability and entailment have been solved using various kinds of automata, such as tree automata~\citep{kupferman:automata-ctl,HladikP06}.

\bibliographystyle{kr}
\bibliography{literature-condensed}

\end{document}